\documentclass{article}
\usepackage{amsfonts}
\usepackage{amssymb}
\usepackage{amsmath}
\usepackage{amsthm}
\usepackage{enumerate}
\usepackage{geometry}
\usepackage{hyperref}
\usepackage{tikz-cd}
\usepackage[utf8]{inputenc}

\newcommand{\df}{\mathrm{d}}

\DeclareMathOperator{\Sym}{Sym}
\DeclareMathOperator{\Hom}{Hom}

\DeclareMathOperator{\Real}{Re}

\DeclareMathOperator{\Int}{Int}
\DeclareMathOperator{\relint}{relint}

\DeclareMathOperator{\Cl}{Cl}
\DeclareMathOperator{\Spec}{Spec}
\DeclareMathOperator{\Conv}{Conv}
\DeclareMathOperator{\pos}{pos}
\DeclareMathOperator{\Log}{Log}

\DeclareMathOperator{\rk}{rk}

\newtheorem{thm}{Theorem}[section]
\newtheorem{lem}[thm]{Lemma}
\newtheorem{col}[thm]{Corollary}
\newtheorem{con}[thm]{Conjecture}
\newtheorem{prop}[thm]{Proposition}
\theoremstyle{definition} 
\newtheorem{defin}[thm]{Definition}
\newtheorem{exam}[thm]{Example}
\theoremstyle{remark}
\newtheorem{remark}[thm]{Remark}

\begin{document}
\date{\today} \title{Toric geometry and regularization of Feynman integrals} \author{Konrad
  Schultka\footnote{schultka@math.hu-berlin.de}}

\maketitle

\begin{abstract}
  We study multivariate Mellin transforms of Laurent polynomials by considering special toric
  compactifications which make their singular structure apparent. This gives a precise description
  of their convergence domain, refining results of Nilsson, Passare, Berkesch and Forsg\aa rd.
  We also reformulate the geometric sector decomposition approach of Kaneko and Ueda in terms of these
  compactifications.

  Specializing to the case of Feynman integrals in the parametric representation, we construct
  multiple such compactifications given by certain systems of subgraphs. As particular cases, we recover
  the sector decompositions of Hepp, Speer and Smirnov, as well as the iterated blow-up
  constructions of Brown and Bloch-Esnault-Kreimer. A fundamental role is played by the Newton
  polytope of the product of the Symanzik polynomials, which we show to be a generalized
  permutahedron for generic kinematics. As an application, we review two approaches to dimensional
  regularization of Feynman integrals, based on sector decomposition and on analytic continuation.
\end{abstract}

\section{Introduction}

The quantum theories of elementary particles are tested by comparing the results of scattering
experiments to the theoretical predictions of the corresponding quantum field theories. The latter
are organized by a sum over Feynman graphs, which encode the possible (virtual) particle histories.
These predictions then amount to the calculation of integrals, which (for $D$-dimensional euclidean
scalar theories) take the form
\begin{equation*}
  I_{G}(q,m) := \int_{H_{1}(G,\mathbb{R}^{D})}\prod_{e\in E_{G}}\frac{1}{(k_{e}+q_{e})^{2}+m^{2}_{e}}\df k,
\end{equation*}
where the integration is over a collection of internal momenta $k\in \mathbb{R}^{D}$, one for each
independent cycle of the graph. These integrals are in general ill-defined. UV divergences present
themselves, when some of the internal momenta become large and the integrand does not fall off
sufficiently rapidly. If some of the edges are massless ($m_{e}=0$), then the integral can also
develop IR-divergences for small momenta, due to non-integrable singularities in the integrand.

Since the work of Weinberg (\cite{Weinberg_1960}) it is known that, at least for generic values of
the external momenta, these singularities are completely described by a simple power-counting
argument for each subgraph $\gamma\subset G$. To make these singularities more apparent, it is often
convenient to work in a parametric representation, which expresses $I_{G}$ as the projective
integral

\begin{equation*}
  I_{G}(q,m) = \Gamma(\omega_{G})\int_{P^{E_{G}}(\mathbb{R}^{+})}\left( \frac{\psi_{G}}{\Phi_{G}} \right)^{\omega_{G}}\frac{\Omega_{P^{E_{G}}}}{\psi_{G}^{D/2}}.
\end{equation*}
Here $P^{E_{G}}$ is the projective space over the set of edges $E_{G}$ and
$\omega_{G}=|E_{G}|-\frac{D}{2}\rk H_{1}(G)$. The polynomials $\psi_{G}$ and $\Phi_{G}$ are the
Symanzik polynomials of the graphs. In this representation, the overall divergence of the graph has
been absorbed in the factor $\Gamma(\omega_{G})$. The UV and IR divergences due to subgraphs
$\gamma\subset G$ now appear because the vanishing loci $V(\psi_{G})$ and $V(\Phi_{G})$ of
$\psi_{G}$ and $\Phi_{G}$ can intersect the boundary of the integration domain
$P^{E_{G}}(\mathbb{R}^{+})\cong \Delta^{E_{G}}$ in the coordinate linear space corresponding to
$\gamma$.

For some applications, this representation is still somewhat inconvenient. Motivated by arithmetic
questions, Bloch, Esnault and Kreimer (\cite{Bloch_2006}) and later Brown (\cite{Brown_2017})
constructed new projective varieties by iteratively blowing-up certain coordinate linear spaces, such
that the strict transforms of $V(\psi_{G})$ and $V(\Phi_{G})$ did not meet the strict transform of
the integration domain. The possible singularities then simply manifest as poles along the
exceptional divisors.

For a completely different purpose, Binoth and Heinrich (\cite{Binoth_2000}) developed an iterative
strategy to decompose the integration domain $P^{E_{G}}(\mathbb{R}^{+})$ into cubical sectors, such
that, after an appropriate coordinate change, the integral over each sector takes the form
\begin{equation}\label{eq:1}
  I_{s} = \int_{[0,1]^{|E_{G}|-1}}x^{\lambda}F(x,q,m)\df x,
\end{equation}
where $F(x,q,m)$ is a rational function which is regular on $[0,1]^{|E_{G}|-1}$. This allowed them
to completely automate the calculation of $I_{G}$ in dimensional regularization. Kaneko and Ueda
(\cite{Kaneko_2010}) later introduced a different, noniterative sector decomposition strategy, which
is based on polyhedral geometry.

The goal of this paper is to show that both construction can be unified in the framework of toric
geometry. A complex variety is toric if it has an action of a complex torus with a dense orbit.
Intuitively, we can think of a toric variety as constructed from a torus by coherently adding tori
of lower dimension at infinity. Normal toric varieties are constructed in terms of a collection of
rational polyhedral cones $\Sigma$, called a fan. The corresponding toric variety $X_{\Sigma}$ is constructed by gluing affine pieces
$U_{\sigma},U_{\sigma'}$ corresponding to cones $\sigma_{1},\sigma_{2}\in\Sigma$ along the open,
dense subvariety $U_{\tau}$ given by the intersection $\tau=\sigma_{1}\cap \sigma_{2}$. This gives a
fruitful interplay between polyhedral and toric geometry. In particular, we can attach a toric
variety $X_{P}$ to every lattice polytope $P\subset \mathbb{Z}^{n}$. Each toric variety $X_{\Sigma}$
also has a natural real, semi-algebraic subset $X_{\Sigma}(\mathbb{R}^{+})\subset
X_{\Sigma}(\mathbb{R})$ called its real-positive locus, which generalizes the simplex
$P^{n}(\mathbb{R}^{+})\cong \Delta^{n}$.

It will be instructive to work in greater generality. We consider general Mellin transforms of
the form
\begin{equation*}
  \mathcal{M}(f_{i},s,c)=\int_{T_{N}(\mathbb{R^{+}})}t^{s}\prod_{i=1}^{k} f_{i}(t)^{-c_{i}}
  \frac{\df t}{t},
\end{equation*}
where $(s,c)\in \mathbb{C}^{n}\times \mathbb{C}^{k}$ are analytic parameters and $f_{i}$ are general
Laurent polynomials on the complex torus $T_{N}\cong (\mathbb{C}^{*})^{n}$, with a suitable
conditions on its coefficients such that the powers $f_{i}^{-c_{i}}$ are well-defined. We prove
that, there is always a smooth complete toric variety $X_{\Sigma}$, such that the closure $V(f_{i})$
does not meet the real-positive locus $X_{\Sigma}(\mathbb{R}^{+})$. We will show that the later is
satisfied if and only if $\Sigma$ refines the normal fans of the Newton polytope of $f=f_{1}\cdots
f_{k}$. In the local coordinates corresponding to a maximal cone $\sigma\in\Sigma$, the integrand of
$\mathcal{M}$ can then be factorized as in Equation \eqref{eq:1}. This allows us to give a precise
characterization of the convergence domain $\Lambda(f_{i})\subset \mathbb{C}^{n}\times
\mathbb{C}^{k}$ of $\mathcal{M}$, generalizing results of Nilsson, Passare, Berkesch and Forsg\aa rd
(\cite{Nilsson_2011},\cite{Berkesch_2014}). For each cone, we can consider the cube
$I_{\sigma}=[0,1]^{n}\subset [0,\infty)^{n}\cong U_{\sigma}(\mathbb{R}^{+})$. These cubes cover
$X_{\Sigma}(\mathbb{R}^{+})$, intersect in a set of measure zero, and in local coordinates, the
integrand has the sector form \eqref{eq:1}. Hence each such toric variety $X_{\Sigma}$ gives a
sector decomposition strategy. In fact, this is just a reformulation of the strategy given by Kaneko
and Ueda.

Coming back to Feynman integrals., it is then natural to promote the integral $I_{G}$ to the Mellin
transform
\begin{equation*}
  I_{G}(\lambda,D,q,m) = \Gamma(\omega_{G})\int_{P^{E_{G}}(\mathbb{R}^{+})}\prod_{i\in E_{G}}\frac{\alpha_{i}^{\lambda_{i}-1}}{\Gamma(\lambda_{i})}
  \left( \frac{\psi_{G}}{\Phi_{G}} \right)^{\omega_{G}}\frac{\Omega_{P^{E_{G}}}}{\psi_{G}^{D/2}}.
\end{equation*}
This corresponds to analytic regularization, i.e. raising each propagator
$(k_{i}^{2}+q_{i}^{2}+m^{2}_{i})^{-1}$ to a power $\lambda_{i}\in \mathbb{C}$.
We will compute the Newton polytope $P_{G}=P(\psi_{G}\Phi_{G})$ for a Feynman graph with generic
euclidean kinematics. In this case, a facet presentation can be easily deduced from the
factorization formulas given by Brown (\cite{Brown_2017}).
As a corollary, we obtain that the
convergence of the Feynman integral is determined by power-counting. It is actually sufficient
to check convergence only for those subgraphs $\gamma\subset G$, such that both $\gamma$ and its quotient
$G/\gamma$ do not contain detachable subgraphs without kinematics. Following Smirnov
(\cite{Smirnov_2012}), we call such graphs s-irreducible.

The polytope $P_{G}$ turns out to be a generalized permutahedron in the sense of (\cite{Postnikov_2009},\cite{aguiar17:hopf}). To such
a polytope, we can canonically attach a smooth refinement of its normal fan, using a combinatorial
version of the wonderful model construction introduced by Feichtner and Kozlov
(\cite{Feichtner_2004}).
For Feynman graphs, this is an iterated blow-up of projective space along coordinate linear spaces
of s-irreducible subgraphs $\gamma\subsetneq G$. The corresponding sectors are the ones constructed
by Smirnov (\cite{Smirnov_2009},[\cite{Smirnov_2012}]). A further refinement then gives the motic
blowup of Brown (\cite{Brown_2017}), which specializes for massive graphs to the original
Bloch-Esnault-Kreimer construction. The toric structure of the later variety has already been used
by Bloch and Kreimer (\cite{Bloch2008}) to interpret renormalization in terms of mixed Hodge
structures. (\cite{Bloch_2006}). We will also reinterpret Speer's sectors \cite{Speer:1975dc} as
given by a certain smooth toric variety, which is not a blowup of projective space in general.

In the last section, we will use our results to give a rigorous construction of dimensional
regularization. We review two approaches, one based on sector decomposition (\cite{Smirnov_1983},
\cite{HEINRICH_2008}, \cite{Bogner_2008}), and the other based on analytic continuation
(\cite{Panzer:2015ida}, \cite{von_Manteuffel_2015}). The results of this section are well-known in
the physics literature, but we included a detailed exposition here since we found it difficult to
find rigorous proofs. Note that this generalizes the construction given by Etingof
(\cite{DEFJKMMRW99}) for massive graphs.

\paragraph{Acknowledgements.}
I thank Dirk Kreimer for encouragement and support. I also greatly benefitted from discussions with Christian Bogner, Henry Ki\ss{}ler
and Marko Berghoff.

\section{Toric varieties}
For the convenience of the reader, we will give a very brief review of the theory of toric
varieties. Our presentation is based on \cite{CLS}, which provides a very comprehensive introduction
to toric geometry.

\paragraph{Cones and fans.}
\label{sec:TV:fans}

Let $N$ be a lattice, i.e. a free abelian group of finite rank $n=\rk N$. Its dual lattice is
$M=\Hom_{\mathbb{Z}}(N,\mathbb{Z})$. The algebraic torus associated to $N$ is
$T_{N}=N\otimes_{\mathbb{Z}}\mathbb{C}^{*}\cong G_{m}^{n}(\mathbb{C})$. Elements $m\in M$ define
characters $t^{m}\in \Hom_{\mathbb{Z}}(T_{N},\mathbb{C}^{*})$. Under this identification, the
coordinate ring of $T_{N}$ is the ring of Laurent polynomials
\begin{equation*}
  \mathcal{O}(T_{N})=\mathbb{C}[M].
\end{equation*}
Similarly, elements $u\in N$ define one-parameter subgroups
\begin{equation*}
  \mathbb{C}^{*}\rightarrow T_{N},\quad \lambda\mapsto u \otimes \lambda.
\end{equation*}

\begin{defin}
  A complex variety $X$ is \emph{toric} if it has an action of a torus $T_{N}$ with a dense torus
  orbit.
\end{defin}

We will only consider \emph{normal} toric varieties. They can be completely described by certain
polyhedral data. The relevant definitions of polyhedral geometry are collected in the appendix.

Let $\sigma\subset N_{\mathbb{R}}:= N \otimes_{\mathbb{Z}} \mathbb{R}$ be a strongly convex
polyhedral cone. $\sigma$ is a \emph{rational} if there are lattice elements $u_{1},\ldots,u_{s}\in
N$, such that
\begin{equation*}
  \sigma = \pos(u_{1},\ldots,u_{s}).
\end{equation*}
Its dual cone
\begin{equation*}
  \sigma^{\vee} := \{m\in M_{\mathbb{R}}\ \rvert\ \langle m, u \rangle\ge 0 \text{ for all } u\in\sigma\}
  \subseteq M_{\mathbb{R}}
\end{equation*}
is again a rational polyhedral cone and the set $S_{\sigma}=M\cap \sigma^{\vee}$ is a finitely
generated semigroup. The toric variety associated to $\sigma$ is
$U_{\sigma}=\Spec(\mathbb{C}[S_{\sigma}])$. The torus action is given in terms of the coordinate
rings by
\begin{equation*}
  \Delta:\mathbb{C}[S_{\sigma}]\rightarrow \mathbb{C}[M]\otimes \mathbb{C}[S_{\sigma}], \quad t^{m}\mapsto
  t^{m}\otimes t^{m}.
\end{equation*}
The inclusion $T_{N}\hookrightarrow X_{\sigma}$ of the dense torus orbit is dually given by the map
$\mathbb{C}[S_{\sigma}]\hookrightarrow \mathbb{C}[M]$.

We can construct all normal toric varieties by gluing such affine varieties along open subsets, as
long as the corresponding cones intersect nicely.

\begin{defin}
  A \emph{fan} $\Sigma$ in $N$ is a collection of rational, strongly convex polyhedral cones such
  that:
  \begin{enumerate}
  \item If $\sigma\in\Sigma$ and $\tau\subseteq \sigma$ is a face of $\Sigma$ then $\tau\in\Sigma$.
  \item Two cones $\sigma_{1},\sigma_{2}\in\Sigma$ are either disjoint or intersect in a common face
    $\tau\in\Sigma$.
  \end{enumerate}
  The set of cones of dimension $k$ is denoted by $\Sigma(k)$. Let
  $|\Sigma|=\bigcup_{\sigma\in\Sigma}\sigma$ be the support of $\Sigma$.

  We call $\Sigma$ a \emph{generalized fan} if it consists of rational polyhedral cones which are
  not necessarily strongly convex.
\end{defin}

If $\sigma_{1},\sigma_{2}$ are two rational, strongly convex cones intersecting in the common face
$\tau=\sigma_{1}\cap \sigma_{2}$, then the dual inclusions $\sigma_{1}^{\vee}\subset
\tau^{\vee}\supset \sigma_{2}^{\vee}$ define the inclusions
\begin{equation*}
  \mathbb{C}[S_{\sigma_{1}}] \hookrightarrow \mathbb{C}[S_{\tau}] \hookleftarrow \mathbb{C}[S_{\sigma_{2}}].
\end{equation*}
One can show that $\mathbb{C}[S_{\tau}]$ is a common localization of $\mathbb{C}[S_{\sigma_{i}}]$,
such that $U_{\tau}\subset U_{\sigma_{i}}$. Gluing $U_{\sigma_{1}}$ and $U_{\sigma_{2}}$ along the
dense open subset $U_{\tau}$ gives a new toric variety. This can be done coherently for all cones in
a fan and we obtain the following:

\begin{prop}[{\cite[Thm. 3.1.5]{CLS}}]
  If $\Sigma$ is a fan in $N_{\mathbb{R}}$ then the $U_{\sigma}$ for $\sigma\in\Sigma$ glue together
  to give a normal toric variety $X_{\Sigma}$ and every normal toric varieties is of this form up to
  isomorphism.
\end{prop}

Properties of $X_{\Sigma}$ are reflected in properties of the fan:
\begin{prop}\label{prop:properties-fan-variety}
  \begin{enumerate}
  \item $X_{\Sigma}$ is complete if $|\Sigma|:=\bigcup_{\sigma\in\Sigma}\sigma=N_{\mathbb{R}}$.
  \item $X_{\Sigma}$ is smooth if and only if every cone $\sigma\in\Sigma$ can be generated by part
    of a $\mathbb{Z}$-basis of $N$.
    In this case the fan $\Sigma$ and its cones $\sigma\in\Sigma$ are called smooth.
  \end{enumerate}
\end{prop}

\begin{proof}
  1. is \cite[Thm. 3.4.1]{CLS} and 2. is \cite[Thm. 3.1.19 ]{CLS}
\end{proof}

\begin{defin}
  Let $X_{\Sigma_{1}},X_{\Sigma_{2}}$ be toric varieties with fans $\Sigma_{i}\subseteq
  (N_{i})_{\mathbb{R}}$. A morphism $\varphi:X_{\Sigma_{1}}\rightarrow X_{\Sigma_{2}}$ is \emph{toric}
  if $\varphi|_{T_{N_{1}}}$ induces a group morphism $T_{N_{1}}\rightarrow T_{N_{2}}$.
\end{defin}

Being toric automatically implies that $\varphi$ is $T_{N_{i}}$-equivariant. We can identify $N_{i}$
with the one-parameter subgroups of $T_{N_{i}}$ and since $\varphi$ is a group morphism, we get a
homomorphism
\begin{equation*}
  \overline\varphi:N_{1}\rightarrow N_{2}
\end{equation*}
of lattices. We say that such a morphism is compatible with the fans $\Sigma_{i}$ if for every
$\sigma_{1}\in\Sigma_{1}$ there is $\sigma_{2}\in\Sigma_{2}$ with
$\overline\varphi(\sigma_{1})\subseteq \sigma_{2}$.

\begin{prop}[{\cite[Thm 3.3.4]{CLS}}]
  If $\varphi:X_{\Sigma_{1}}\rightarrow X_{\Sigma_{2}}$ is toric, then the induced map
  $\varphi:N_{1}\rightarrow N_{2}$ is compatible with the fans $\Sigma_{1},\Sigma_{2}$.

  Conversely, every morphism $\overline\varphi:N_{1}\rightarrow N_{2}$ compatible with the fans
  uniquely determines a toric morphism $\varphi:X_{\Sigma_{1}}\rightarrow X_{\Sigma_{2}}$ which extends
  \begin{equation*}
    \overline\varphi \otimes 1:N_{1} \otimes \mathbb{C}^{*}=
    T_{N_{1}}\rightarrow N_{2} \otimes \mathbb{C}^{*}=T_{N_{2}}.
  \end{equation*}
\end{prop}
\begin{remark}\label{rem:proper-toric-morphism}
  One can show that $\varphi:X_{\Sigma_{1}}\rightarrow X_{\Sigma_{2}}$ is proper if and only if
  $\varphi^{-1}(|\Sigma_{2}|)=|\Sigma_{1}|$. See \cite[Thm 3.4.11]{CLS}.
\end{remark}

\begin{exam}
  Suppose $\varphi$ is the identity and $\Sigma_{1}$ is a refinement of $\Sigma_{2}$, i.e.
  $|\Sigma_{1}| = |\Sigma_{2}|$ and every cone $\sigma_{1}\in \Sigma_{1}$ is contained in some cone
  $\sigma_{2}\in\Sigma_{2}$. Then the corresponding map $X_{\Sigma_{1}}\rightarrow X_{\Sigma_{2}}$
  is proper and birational.
\end{exam}

For some applications, it is useful to work with the following weakened version of smoothness.
\begin{defin}
  A strongly convex polyhedral cone $\sigma$ is called \emph{simplicial} if its generators are
  linearly independent. A toric variety $X_{\Sigma}$ is \emph{simplicial} if its fan $\Sigma$
  consists of simplicial cones.
\end{defin}

We will see below that a simplicial variety has only abelian quotient singularities. For many
purposes, this is as good as smoothness.

\paragraph{The orbit-cone correspondence.}
\label{sec:TV:Orbit-Cone}

A cone $\sigma\in\Sigma$ defines a distinguished point $\gamma_{\sigma}\in U_{\sigma}\subseteq
X_{\Sigma}$: $\gamma_{\sigma}$ is given by the semigroup morphism:
\begin{equation*}
  m\in S_{\sigma}\mapsto
  \begin{cases}
    1 ,\quad m\in \sigma^{\bot}\cap M\\
    0 ,\quad \text{otherwise}
  \end{cases}
\end{equation*}
This is a fixed point for the $T_{N}$ action if and only if $\dim\sigma=\dim N_{\mathbb{R}}$.

\begin{thm}[Orbit-Cone Correspondence]\label{sec:orbit-cone-corr}
  There is a bijective correspondence
  \begin{align*}
    \{\sigma\in\Sigma \} &\longleftrightarrow \{T_{N}\text{-orbits}\subseteq \Sigma\} \\
    \sigma &\longleftrightarrow O(\sigma):= T_{N}\cdot \gamma_{\sigma}
  \end{align*}
  having the following properties:
  \begin{enumerate}
  \item $\dim \sigma + \dim O(\sigma)=\dim N_{\mathbb{R}}$
  \item The affine open set $U_{\sigma}$ decomposes into orbits as
    \begin{equation*}
      U_{\sigma}=\bigcup_{\tau\preceq\sigma}O(\tau)
    \end{equation*}
  \item $\tau\preceq\sigma$ if and only if $O(\sigma)\subseteq \overline{O(\tau)}$, and
    \begin{equation*}
      V(\tau):=\overline{O(\tau)} = \bigcup_{\tau\preceq\sigma}O(\sigma)
    \end{equation*}
  \end{enumerate}
\end{thm}
\begin{proof}
  \cite[Theorem 3.2.6 and Prop. 3.2.7]{CLS}
\end{proof}

\paragraph{Divisors and the homogeneous coordinate ring.}
\label{sec:TV:divisor-coordinate-ring}
Let $X_{\Sigma}$ be a toric variety associated to the fan $\Sigma$. A one-dimensional cone
$\rho\in\Sigma(1)$ gives a torus-invariant divisor $D_{\rho}=V(\rho)$ under the orbit-cone
correspondence and every torus-invariant divisor is a sum of these. Denoting the latter group by
$Div_{T}(X_{\Sigma})$, we have an identification
\begin{equation*}
  \mathbb{Z}^{\Sigma(1)}\cong Div_{T}(X_{\Sigma}).
\end{equation*}
Since $\rho$ is a one-dimensional rational cone, there is a unique smallest lattice generatore of
$\rho$, i.e. an element $u_{\rho}\in N$ such that $\rho=\mathbb{R}^{+}u_{\rho}$.

\begin{prop}[{\cite[Thm 4.1.3]{CLS}}]\label{prop:divis-toric-vari}
  There is an exact sequence
  \begin{center}
    \begin{tikzcd}
      M \arrow{r}{} & \mathbb{Z}^{\Sigma(1)} \arrow{r}{} & \Cl(X_{\Sigma{}}) \arrow{r} & 0
    \end{tikzcd}
  \end{center}
  where $\Cl(X_{\Sigma})$ denotes the class group. The first morphism maps $m\in M$ to the divisor
  \begin{equation*}
    div(t^{m}) = \sum_{\rho\in\Sigma(1)}\langle m, u_{\rho} \rangle D_{\rho}
  \end{equation*}
  of the rational function $t^{m}$. The second is the natural quotient map
  $\mathbb{Z}^{\Sigma(1)}\cong Div_{T}(X_{\Sigma})\rightarrow \Cl(X_{\Sigma})$. If $X_{\Sigma}$ has
  no torus factors, i.e. it is not of the form $X_{\Sigma}\cong X_{\Sigma'}\times T^{k}$, then there
  is a short exact sequence
  \begin{center}
    \begin{tikzcd}
      0 \arrow{r} & M \arrow{r}{} & \mathbb{Z}^{\Sigma(1)} \arrow{r}{} & \Cl(X) \arrow{r} & 0.
    \end{tikzcd}
  \end{center}
\end{prop}
The global sections of torus-invariant divisors are described by polyhedra as follows (\cite[Prop.
4.3.3]{CLS}): If $D=\sum_{\rho}a_{\rho}D_{\rho}$ is a torus-invariant divisor on $X_{\Sigma}$, then
\begin{equation*}
  \Gamma(X_{\Sigma},\mathcal{O}_{X_{\Sigma}}(D)) = \bigoplus_{m\in P_{D}\cap M}\mathbb{C}\cdot t^{m},
\end{equation*}
where
\begin{equation*}
  P_{D} := \{m\in M_{\mathbb{R}}\ \rvert\ \langle m, u_{\rho} \rangle\ge -a_{\rho}\}.
\end{equation*}
  
Now suppose $X_{\Sigma}$ is a simplicial toric variety without torus factors. We want to describe
$X_{\Sigma}$ by a graded ring $S_{\Sigma}$, generalizing the homogeneous coordinate description of
projective space. Tensoring the exact sequence of the class group with
$\Hom_{\mathbb{Z}}(-,\mathbb{C}^{*})$ gives the exact sequence
\begin{center}
  \begin{tikzcd}
    1 \arrow{r} & G \arrow{r}{} & (\mathbb{C}^{*})^{\Sigma(1)}\arrow{r} & T_N \arrow{r} & 1
  \end{tikzcd}
\end{center}
where $G=\Hom_{\mathbb{Z}}(\Cl(X),\mathbb{C}^{*})$ is the character group of $\Cl(X)$. This is a
reductive group isomorphic to the product of a torus and a finite group. We can describe $G$
concretely as
\begin{equation*}
  G = \{(t_{\rho})\in
  (\mathbb{C}^{*})^{\Sigma(1)}\ \rvert\ \prod_{\rho}t_{\rho}^{\langle m,
    u_{\rho} \rangle}=1 \text{ for all } m\in M\}.
\end{equation*}

\begin{defin}
  The \emph{homogeneous coordinate ring} of $X_{\Sigma}$ is
  \begin{equation*}
    S_{\Sigma}=\mathbb{C}[x_{\rho}\ \rvert\ \rho\in\Sigma(1)]=\mathcal{O}(\mathbb{C}^{\Sigma(1)}).
  \end{equation*}
\end{defin}

The ring $S_{\Sigma}$ is graded by $\Cl(X)$:
\begin{equation*}
  \deg(x_{\rho}) = [D_{\rho}]
\end{equation*}
This gives an action of $G$ by duality, which is just the restriction of the natural scaling action
of $(\mathbb{C}^{*})^{\Sigma(1)}$. The corresponding eigenspaces are the graded components of
$S_{\Sigma}$:
\begin{equation*}
  S_{\Sigma} = \bigoplus_{\beta\in \Cl(X)}S_{\beta}.
\end{equation*}
We want to describe $X_{\Sigma}$ as a suitable quotient of
$\Spec(S_{\Sigma})=\mathbb{C}^{\Sigma(1)}$ by $G_{}$. For this to work, we first have to throw out
some badly behaved $G$-orbits.
\begin{defin}
  For $\sigma\in\Sigma$ let $x^{\hat\sigma}=\prod_{\rho\notin \sigma(1)}x_{\rho}\in S_{\Sigma}$. The
  \emph{irrelevant ideal} is
  \begin{equation*}
    B_{\Sigma}=\langle x^{\hat\sigma}\ \rvert\  \sigma\in\Sigma_{max}\rangle.
  \end{equation*}
  The corresponding zero set $Z_{\Sigma}=V(B_{\Sigma})$ is the \emph{irrelevant locus}.
\end{defin}

The variety $\mathbb{C}^{\Sigma(1)}\backslash Z_{\Sigma}$ is again toric. Its fan can be described
as follows: For $\sigma\in\Sigma$ let
\begin{equation*}
  \tilde\sigma = \pos(e_{\rho}\ \rvert\ \rho\in\sigma)\subseteq \mathbb{R}^{\Sigma(1)}.
\end{equation*}
The collection of all $\tilde\sigma$ constitute the fan $\tilde\Sigma$ of
$\mathbb{C}^{\Sigma(1)}\backslash Z$. The lattice morphism
\begin{align*}
  \overline \pi:\mathbb{Z}^{\Sigma(1)}\rightarrow N,\quad e_{\rho}\mapsto u_{\rho}
\end{align*}
is obviously compatible with the fans. Hence we get a toric morphism
\begin{equation*}
  \pi:\mathbb{C}^{\Sigma(1)}\backslash Z_{\Sigma}\rightarrow X_{\Sigma}.
\end{equation*}

\begin{thm}[{\cite[Thm. 5.1.11]{CLS}}]\label{thm:geometric-quotient}
  Let $X_{\Sigma}$ be a simplicial toric variety without torus factors. The map $\pi$ describes
  $X_{\Sigma}$ as the geometric quotient
  \begin{equation*}
    X_{\Sigma}=\mathbb{C}^{\Sigma(1)}\backslash Z_{\Sigma}//G
  \end{equation*}
\end{thm}

\begin{remark}
  We refer to \cite[Section 5.0]{CLS} for background on geometric invariant theory and geometric
  quotients. Let us point out some consequences of this result:
  \begin{enumerate}
  \item The $G$-orbits on $\mathbb{C}^{\Sigma(1)}\backslash Z_{\Sigma}$ are closed and the closed
    points of $X_{\Sigma}$ is the orbit space.
  \item For an affine open subset $U=\Spec(R)\subseteq X_{\Sigma}$ we have $\pi^{-1}(U)=\Spec(\tilde
    S)$, where $\tilde S$ is a localization of $S_{\Sigma}$ with an induced $G$-action.
    That $\pi$ is a geometric quotient implies that $R=\tilde S^{G}$, i.e. $R$ is the subring of
    $G$-invariants.
  \end{enumerate}
\end{remark}

Let us specialize the above remark to an affine open $U_{\sigma}\subseteq X_{\Sigma}$ given by a
cone $\sigma\in\Sigma$. For the inverse image we have
\begin{align*}
  \pi^{-1}(U_{\sigma}) = U_{\tilde\sigma}=\Spec(\mathbb{C}[\tilde \sigma^{\vee}\cap M])
\end{align*}
where $\tilde\sigma=\pos(e_{\rho} \ \rvert\ \rho\in\sigma(1)).$ The coordinate ring is then
\begin{equation*}
  \mathbb{C}[\tilde \sigma^{\vee}\cap M]
  =\mathbb{C}\left[\prod_{\rho}x_{\rho}^{a_{\rho}}\ \rvert\ a_{\rho}\ge 0 \text{ for } \rho\in\sigma(1)\right]
  := S_{x^{\hat\sigma}},
\end{equation*}
i.e. we invert every variable $x_{\rho}$ for $\rho\notin \sigma(1)$. Hence we get
\begin{equation*}
  \pi^{-1}(U_{\sigma})=\Spec(S_{x^{\hat\sigma}}).
\end{equation*}
The map on coordinate rings is given by homogenization:
\begin{align*}
  \pi^{*}:\mathbb{C}[\sigma^{\vee}\cap M]&\longrightarrow S_{x^{\hat\sigma}}\\
  \pi^{*}(t^{m})&=\prod_{\rho}x^{\langle m, u_{\rho} \rangle}_{\rho}
\end{align*}
Its image is the space of $G$-invariants $S^{G}_{x^{\hat\sigma}}$. This gives the isomorphism
\begin{equation*}
  U_{\sigma}=\Spec( \mathbb{C}[\sigma^{\vee}\cap M])\cong \Spec(S^{G}_{x^{\hat\sigma}}).
\end{equation*}
For top-dimensional cones $\sigma\in\Sigma(\dim N)$, we can describe the above isomorphism by
dehomogenization, i.e. setting some of the variables $x_{\rho}$ to 1. Let
\begin{align*}
  \varphi_{\sigma}&:\mathbb{C}^{\sigma(1)}\rightarrow \mathbb{C}^{\Sigma(1)} \\
  \varphi_{\sigma}(a)_{\rho} &=
                            \begin{cases}
                              a_{\rho}, &\quad \rho\in\sigma(1) \\
                              1, &\quad \rho\notin\sigma(1).
                            \end{cases}
\end{align*}
The diagram
\begin{center}\begin{tikzcd}
    \mathbb{C}^{\sigma(1)}\arrow[hook]{r}{\varphi_\rho}\arrow{d}{\pi_{\sigma{}}} & \mathbb{C}^{\Sigma(1)}\backslash Z(\Sigma{}) \arrow{d}{\pi_{\Sigma{}}} \\
    U_\sigma{} \arrow[hook]{r} & X_\Sigma{}
  \end{tikzcd}\end{center}
commutes and the left vertical map is an isomorphism if $\sigma$ is smooth. For simplicial cones, we
still have the following:
\begin{prop}
  The map $\pi_{\sigma}:\mathbb{C}^{\sigma(1)}\rightarrow U_{\sigma}$ is the geometric quotient of
  $\mathbb{C}^{\sigma(1)}$ by the finite group
  $G_{\sigma}=\Hom_{\mathbb{Z}}(\Cl(U_{\sigma}),\mathbb{C}^{*})$.
\end{prop}
\begin{proof}
  This is a special case of Theorem \ref{thm:geometric-quotient}. We have $Z_{\sigma}=\emptyset$ and
  the class group is the quotient
  \begin{equation*}
    \Cl(U_{\sigma}) = \mathbb{Z}^{\sigma(1)}/ M.
  \end{equation*}
  This group is torsion and thus finite, since the ray generators of $\sigma(1)$ furnish an
  isomorphism $\mathbb{R}^{\sigma(1)}\cong N_{\mathbb{R}}$. The character group $G_{\sigma}$ must
  then also be finite.
\end{proof}
\begin{col}
  A simplicial toric variety without torus factors $X_{\Sigma}$ has a natural orbifold structure
  with abelian local groups.
\end{col}

There is a general correspondence between graded $S_{\Sigma}$-modules and quasi-coherent sheaves on
$X_{\Sigma}$. We want to explain this in the case of the canonical sheaf, which is defined as
follows: Let $j:X^{sm}_{\Sigma}\rightarrow X_{\Sigma}$ be the inclusion of the smooth part of
$X_{\Sigma}$. The canonical sheaf is then
\begin{equation*}
  \omega_{\Sigma} = j_{*}\Omega^{n}_{X^{sm}_{\Sigma}},
\end{equation*}
where $\Omega^{n}_{X_{\Sigma}^{sm}}$ is the sheaf of holomorphic $n$-forms on $X^{sm}_{\Sigma}$. In
other words, a section of $\omega_{\Sigma}$ over an open $U\subseteq X_{\Sigma}$ is given by a
holomorphic $n$-form on $U^{sm}=U \cap X^{sm}$.

Let $(e_{1},\ldots,e_{n})$ be a basis of $M$ and $I=\{\rho_{1},\ldots,\rho_{n}\}\subseteq \Sigma(1)$
an $n$-element subset. Let $u_{I}=\det(\langle e_{i}, u_{\rho_{j}} \rangle_{_{ij}})$ and set
\begin{equation*}
  \Omega_{\Sigma} = \sum_{I}u_{I}\left( \prod_{\rho\notin I}x_{\rho} \right)
  dx_{\rho_{1}}\wedge\ldots\wedge dx_{\rho_{n}}.
\end{equation*}
This is an element of the $S_{\Sigma}$-module
\begin{equation*}
  \bigwedge^{n}\Omega^{1}_{S_{\Sigma}} \cong \Gamma(\mathbb{C}^{\Sigma(1)},\Omega^{n}_{\mathbb{C}^{\Sigma(1)}}),
\end{equation*}
where $\Omega^{1}_{S_{\Sigma}}$ is the module of K\"ahler differentials over $S_{\Sigma}$. Note that
$\Omega_{X_{\Sigma}}$ is independent of the above choices up to sign (i.e. up to the choice of an
orientation of $M$ or $N$). The group $G$ acts on $\bigwedge \Omega^{1}_{\Sigma_{\Sigma}}$ by
pullback. The action of $t^{m}\in G\subseteq (\mathbb{C}^{*})^{\Sigma(1)}$ is given by
\begin{equation*}
  t^{m}\cdot \Omega_{\Sigma} = t^{\langle m, \sum_{\rho}u_{\rho}\rangle}\Omega_{\Sigma}.
\end{equation*}
Hence $\Omega_{\Sigma}$ has degree
\begin{equation*}
  \beta = \left[\sum_{\rho}D_{\rho}\right]\in \Cl(X_{\Sigma}).
\end{equation*}
If $F,H\in S_{\Sigma}$ are polynomials, such that $\deg F-\deg H=-\beta$, then the meromorphic
$n$-form
\begin{equation*}
  \frac{F(x)}{H(x)}\Omega_{\Sigma}
\end{equation*}
is $G$-invariant and descends to a global meromorphic section of $\omega_{\Sigma}$. Conversely, let
$f,h\in \mathcal{O}(T_{N})$ be Laurent polynomials and consider the section
\begin{equation*}
  \alpha = \frac{f(t)}{h(t)}\frac{dt_{1}}{t_{1}}\wedge \ldots \wedge \frac{dt_{n}}{t_{n}} \in \Gamma(T_{N},\omega_{T_{N}}),
\end{equation*}
where $t_{1},\ldots,t_{n}$ are generators of $\mathcal{O}(T_{N})$. Pulling back along the quotient
map $\pi: \mathbb{C}^{\Sigma(1)} \backslash Z_{\Sigma}\rightarrow X_{\Sigma}$ gives a meromorphic
form $\pi^{*}\alpha$ on $\mathbb{C}^{\Sigma(1)}$, which we can describe as follows.
\begin{prop}\label{prop:differential-form}
  Let $F=\pi^{*}f$ and $H=\pi^{*}h\prod_{\rho}x_{\rho}$. Then the pullback $\pi^{*}\alpha$ is given
  by
  \begin{equation*}
    \pi^{*}\alpha = \frac{F(x)}{H(x)}\Omega_{\Sigma}.
  \end{equation*}
\end{prop}

\begin{proof}
  From $\pi^{*}(t^{m})=\prod_{\rho}x_{\rho}^{\langle m, u_{\rho} \rangle}$ it follows that
  \begin{equation*}
    \pi^{*}\left( \frac{dt_{1}}{t_{1}} \right) = \sum_{\rho}\langle e_{i}, u_{\rho} \rangle \frac{dx_{\rho}}{x_{\rho}}.
  \end{equation*}
  Taking the wedge product and multiplying by $\prod_{\rho}x_{\rho}$ gives
  \begin{equation*}
    \left( \prod_{\rho}x_{\rho} \right)\pi^{*}\left( \frac{dt_{1}}{t_{1}}\wedge \ldots \wedge \frac{dt_{n}}{t_{n}} \right) = \Omega_{\Sigma}
  \end{equation*}
  and the above formula follows.
\end{proof}

Over the maximal cone $\sigma=\pos(u_{i}\ \rvert\ i\in I_{0})$, we formally set $x_{\rho}=1$ for
$\rho\notin\sigma(1)$ and get the section
\begin{equation*}
  \alpha\big\rvert_{U_{\sigma}} = u_{\sigma}\frac{\tilde f(x)}{\tilde g(x)}\df x_{i_{1}}\wedge \ldots \wedge \df x_{i_{n}}\in \Gamma(U_{\sigma},\omega_{\Sigma}),
\end{equation*}
where $u_{\sigma}=u_{I_{0}}$ and $\tilde f(x)$ and $\tilde h(x)$ are the dehomogenizations of $F$
and $H$, i.e. the restriction of $F$ and $H$ to the set $\{x_{\rho}=1 \ \rvert\
\rho\notin\sigma(1)\}$.

\paragraph{Lattice polytopes.}
\label{sec:TV:Polytopes}

Now let $P \subseteq M_{\mathbb{R}}$ be a full-dimensional lattice polytope, i.e. it is the convex
hull of finitely many lattice points. Denote by $P(k)$ the set of $k$-dimensional faces. $P$ can be
described by the facet presentation
\begin{equation*}
  P = \{m\in M| \langle m,u_{F}\rangle \ge -a_{F}, F\in P(n-1)\},
\end{equation*}
where $a_{F}\in \mathbb{Z}$ and $u_{F}$ is the minimal lattice generator of the cone $\rho_{F}$
consisting of inward pointing normal vectors of $F$.

We can construct the \emph{normal fan} of $P$ as follows: Let $v\in P$ be a vertex and $C_{v}$ the
cone generated by $P\cap M -v$. Its dual cone $\sigma_{v}=C_{v}^{\vee}$ is again rational and
strongly convex. In terms of the facet presentation above, we have
\begin{equation*}
  \sigma_{v} = \pos(u_{F}\ \rvert\  F\in P(n-1), v\in F).
\end{equation*}
More generally, we can define for any face $Q\in P(k)$ the cone
\begin{equation*}
  \sigma_{Q} = \pos(u_{F}\ \rvert\  F\in P(n-1), Q\subseteq F).
\end{equation*}

\begin{prop}\label{prop:normal-fan-polytope}
  The cones $\sigma_{Q}$ constitute a complete fan $\Sigma_{P}$ in $N_{\mathbb{R}}$ and define a
  complete toric variety $X_{P}:=X_{\Sigma_{P}}$. A vector $u\in N_{\mathbb{R}}$ defines the face
  $F_{u}P=Q$ if and only if $u\in \relint(\sigma_{Q})$. This defines an inclusion reversing
  bijection between $\Sigma_{P}$ and the set of faces of $P$.
\end{prop}
\begin{proof}
  See \cite[Prop. 2.3.7 and Prop. 3.1.6]{CLS}.
\end{proof}
If $P$ is not full-dimensional, then the same construction will give a generalized fan.

\begin{remark}
  The Orbit-Cone correspondence takes the following form: The cones $\sigma\in \Sigma_{P}$
  correspond to faces $Q\subseteq P$, hence every $Q$ gives a torus orbit $O(Q)\subseteq \Sigma_{P}$
  and its closure $V(Q)$. The latter is a closed toric subvariety and hence again given by a
  complete fan, which we can describe as follows: By translating $Q$ by one of its vertices we can
  assume that $0\in Q$. Let then $M_{Q}$ be the linear span of $Q$ such that $Q\subseteq M_{Q}$
  becomes a full-dimensional lattice polytope which has a normal fan $\Sigma_{Q}$. Then we have
  $V(Q)\cong X_{Q}$. See \cite[Prop 3.2.9]{CLS}.

  Identifying rays $\rho\in\Sigma_{P}(1)$ with facets $F$ of $P$, we get the divisor
  $D_{P}=\sum_{F}a_{F}D_{F}$ canonically attached to $X_{P}$. One can show that $D_{P}$ is ample and
  there is a bijective correspondence
  \begin{equation*}
    P \longleftrightarrow (X_{\Sigma},D)
  \end{equation*}
  between full dimensional lattice polytopes $P \subseteq M_{\mathbb{R}}$ and complete toric
  varieties $X_{\Sigma}$ with fan $\Sigma\subseteq N_{\mathbb{R}}$ together with a distinguished
  torus-invariant ample divisor $D$. See \cite[Thm. 6.2.1]{CLS}.
\end{remark}

The following proposition will be needed later.
\begin{prop}\label{prop:maximal-cones}
  Suppose $P\subset M$ is a (not necessarily full-dimensional) lattice polytope. An $n$-dimensional
  cone $\sigma=\pos(u_{1},\ldots,u_{n})$ is contained in a maximal cone $\tilde
  \sigma\in\Sigma_{P}(n)$ if and only if there is a unique vertex $m_{\sigma}\in P(f)$, such that
  \begin{equation*}
    \langle m_{\sigma}, u_{i} \rangle = \min_{m\in P}\langle m, u_{i} \rangle
  \end{equation*}
  for all $i=1,\ldots,n$.
\end{prop}

\begin{proof}
  Let $\tilde \sigma\in\Sigma_{P}(n)$, correspond to the vertex $m_{0}\in P(0)$. The
  cone $\sigma=\pos(u_{1},\ldots, u_{n})$ is contained in $\tilde \sigma$ if and only if every
  weight vector $w=\sum_{i=1}^{n}\lambda_{i}u_{i}\in \Int(\sigma)$ defines the face
  $F_{w}P=\{m_{0}\}$. This means that
  \begin{equation*}
    \langle m_{0}, w \rangle < \langle m, w \rangle, \text{ for all } m\in P \backslash\{m_{0}\}.
  \end{equation*}
  Varying $\lambda_{i}$, it is easy to see that this is only possible if
  \begin{equation*}
    \langle m_{0}, u_{i} \rangle = \min_{m\in A}\langle m, u_{i} \rangle,
  \end{equation*}
  for all $i=1,\ldots,n$. Conversely, suppose $m_{0}$ minimizes $\langle m, u_{i} \rangle$ for
  all $i$ and thus for all $w\in\sigma$. Suppose there is another $\tilde m\in P$, such that
  $\langle \tilde m, w \rangle$ is minimal. Then $\langle\tilde m-m_{0}, w \rangle=0$ and $\langle
  \tilde m-m_{0}, v\rangle\ge 0$ for all other $v\in\sigma$. This means that $w$ lies in a face of $\sigma$
  and thus $m\notin\Int(\sigma)$.
\end{proof}

For two lattice polytopes $P_{1},P_{2}\subset M_{\mathbb{R}}$, let
\begin{equation*}
  Q = P_{1} + P_{2} = \{m_{1}+m_{2} \ \rvert\ m_{1}\in P_{1}, m_{2}\in P_{2}\}
\end{equation*}
be their Minkowski sum. This is clearly a lattice polytope again.

\begin{prop}[{\cite[Prop. 7.12]{Ziegler_1995}}]\label{prop:common-refinement}
  The normal fan $\Sigma_{Q}$ of $Q$ is the coarsest common refinement of the normal fans
  $\Sigma_{P_{1}},\Sigma_{P_{2}}$.
\end{prop}

\paragraph{Real-positive locus.}
In the sequel, we want to integrate holomorphic forms over a compact $n$-cycle, which is naturally
associated to every complete toric variety. Let $\sigma\in\Sigma$ be a cone in the fan defining
$X_{\Sigma}$. The complex points of the corresponding affine toric variety are given by
\begin{equation*}
  U_{\sigma}(\mathbb{C})=\Hom(\sigma^{\vee}\cap M,\mathbb{C}).
\end{equation*}
Restricting the image to $\mathbb{R}^{+}$ gives the locus $U_{\sigma}(\mathbb{R}^{+})$. These glue
together to give the real positive locus $X_{\Sigma}(\mathbb{R}^{+})$. A toric morphism
$X_{\Sigma}\rightarrow X_{\tilde{\Sigma}}$ induces a map $X_{\Sigma}(\mathbb{R}^{+})\rightarrow
X_{\tilde{\Sigma}}(\mathbb{R}^{+})$.


\begin{exam}
  Suppose $X_{\Sigma}$ is a projective toric variety associated to the polytope $P$, such that the
  divisor $D_{P}$ is very ample. Its sections
  \begin{equation*}
    t^{m_{i}}\in\Gamma(X_{\Sigma},\mathcal{O}_{X_{P}}(D_{P})),\quad m_{i}\in P\cap M
  \end{equation*}
  furnish a projective embedding
  \begin{equation*}
    X_{\Sigma}\rightarrow \mathbb{P}^{s},\quad x\mapsto [t^{m_{0}}(x):\ldots:t^{m_{s}}(x)].
  \end{equation*}
  The (algebraic) moment is defined as
  \begin{align*}
    f:X_{\Sigma}&\rightarrow M_{\mathbb{R}}\\
    f(x) &= \frac{\sum_{m\in P\cap M}|t^{m}(x)|m}{\sum_{m\in P\cap M}|t^{m}(x)|}.
  \end{align*}
  By \cite[Thm. 12.2.5]{CLS}, this induces a homeomorphism
  \begin{equation*}
    f:X_{\Sigma}(\mathbb{R}^{+})\tilde{\longrightarrow} P,
  \end{equation*}
  which identifies a facet $Q\subseteq P$ with $V(Q)\cap X_{\Sigma}(\mathbb{R}^{+})$.
\end{exam}

\paragraph{Star subdivision of fans.}
\label{sec:TV:blowup}
There is a standard construction to refine a given fan $\Sigma$. Let $\nu\in \Sigma\cap N$ be a
primitive element, i.e. such that $\nu$ is the lattice generator of $\pos(\nu)$. For
$\sigma\in\Sigma$ with $\nu\in\sigma$ let
\begin{equation*}
  \Sigma_{\sigma}(\nu) = \{\pos(\tau,\nu)\ \rvert\ \{\nu\}\cup\tau\subseteq \sigma, \nu\notin \tau\}.
\end{equation*}
The \emph{star subdivision} of $\Sigma$ with respect to $\nu$ is the fan
\begin{equation*}
  \Sigma^{*}(\nu) = \{\sigma\in\Sigma \ \rvert\ \nu \notin\sigma\}\cup\bigcup_{\nu\in\sigma}\Sigma_{\sigma}(\tau).
\end{equation*}
The identity map $N\rightarrow N$ is compatible with the fans $(\Sigma^{*}(\nu),\Sigma)$ and induces
a toric morphism
\begin{equation*}
  \pi:X_{\Sigma^{*}(\nu)}\rightarrow X_{\Sigma},
\end{equation*}
which is proper and birational.

We are interested in two special cases of this construction. First let $X_{\Sigma}$ be a smooth
toric variety associated to the fan $\Sigma$. From Prop. \ref{prop:properties-fan-variety} we know
that every cone $\sigma\in \Sigma$ is smooth, i.e. can be generated by part of a $\mathbb{Z}$-basis
of $N$. For a cone $\tau\in\Sigma$, the closure $V(\tau)=\overline O(\tau)$ is a smooth toric
subvariety. Let $\Sigma_{\tau}$ be the star subdivision of $\Sigma$ with respect to the vector
\begin{equation*}
  \nu_{\tau} = \sum_{\rho\in\tau(1)}u_{\rho}
\end{equation*}

\begin{prop}[{\cite[Prop. 1.26]{Oda88}}]
  The map $\pi:X_{\Sigma_{\tau}}\rightarrow X_{\Sigma}$ is the blow-up of $X_{\Sigma}$ with center
  $V(\tau)$.
\end{prop}

Now let $\Sigma$ be a not necessarily simplicial fan and $\nu_{\rho}$ be the ray generator of a cone
$\rho$. A non-simplicial cone $\sigma\in\Sigma$ containing $\rho$ gets subdivided in
$\Sigma^{*}(\nu_{\rho})$. Iterating this construction gives the following:

\begin{thm}[{\cite[Prop. 11.1.7]{CLS}}]\label{thm:simplicial-refinement}
  Let $\Sigma$ be non-simplicial fan. Then there is a fan $\Sigma'$, obtained from $\Sigma$ by a
  series of star subdivisions in rays $\rho_{1},\ldots,\rho_{s}\in\Sigma(1)$, such that $\Sigma'$ is
  simplicial and $\Sigma'(1)=\Sigma(1)$.
\end{thm}

\paragraph{Toric wonderful models.}
\label{sec:TV:iterblowup}
In this section, we want to describe certain compactifications of the torus $T^{n-1}$ given by
iteratively blowing up coordinate subspaces in the projective compactification
$T^{n-1}\hookrightarrow P^{n-1}$. These are special cases of the wonderful model compactifications
of \cite{De_Concini_1995}.

Let $E$ be a finite set with $n$ elements and $P^{E_{G}}$ be the projective space of dimension $n-1$,
where we label the homogeneous coordinates by elements of $E$. Let
\begin{equation*}
  N_{E}=\mathbb{Z}^{E}/\mathbb{Z}\left(\sum_{i\in E}e^{i}\right)\cong \mathbb{Z}^{n-1}
\end{equation*}
and
\begin{equation*}
  M_{E} = \left\{ m\in \mathbb{Z}^{E}| \sum_{i\in E}m_{i}=0 \right\} 
\end{equation*}
the dual lattice. The fan $\Sigma_{E}$ of $P^{E_{G}}$ is given by the cones
\begin{equation*}
  \tau_{I} := \pos([e^{i}]\ \rvert\ i\in I)
\end{equation*}
for all $I\subsetneq E$. Every such proper subset $I\subsetneq E$ then gives the linear subspace
\begin{equation*}
  L_{I} = \{[\alpha_{j}]\ \rvert\ \alpha_{i}=0 \text{ for } i\in I\}\cong P^{I^{c}},
\end{equation*}
which is the orbit closure associated to the cone $\tau_{I}$.

Consider a set of subsets $B\subseteq 2^{E}$ satisfying the following conditions.
\begin{enumerate}
\item $E\notin B$
\item $\{i\}\notin B$ for all $i\in E$.
\item $I_{1},I_{2}\in B, I_{1}\cap I_{2}\neq \emptyset \Rightarrow I_{1}\cup I_{2}\in B$.
\end{enumerate}
The iterated Blow-up
\begin{equation*}
  \pi_{B}:P^{B}\rightarrow P^{E_{G}}
\end{equation*}
is defined by inductively blowing up the elements of
\begin{equation*}
  \mathcal{L}_{B} = \{L_{I}\ \rvert\ I\in B\},
\end{equation*}
in order of increasing dimension. More precisely, let $B=\{I_{1}<\ldots<I_{m}\}$ be linearly ordered
such that $j\ge k$ implies that $I_{j}\subseteq I_{k}$. We then define the sequence of blow-ups by
$P_{0}=P^{E_{G}}$ and $P_{k}=Bl_{\tilde L_{I_{k}}}P_{k-1}$, where $\tilde L_{k}$ is the strict transform
of $L_{k}$ in $P_{k-1}$. The results of the last section show that this is again a smooth projective
toric variety. Let $\Sigma_{k}$ be the fan of $P_{k}$. Then we have that $\Sigma_{k} =
St_{\tau_{k}}\Sigma_{k-1}$ is the star subdivision with respect to the cone
\begin{equation*}
  \tau_{k} = \pos\{[e^{i}]\ \rvert\ i\in I_{k}\}.
\end{equation*}
Blowing up in order of increasing dimension ensures that this is well-defined, i.e. that $\tau_{k}$
is indeed a cone of $\Sigma_{k-1}$. The strict transform of $L_{I_{k}}$ in $P_{k-1}$ is just the
orbit closure $V(\tau_{k})$ of $\tau_{k}\in\Sigma_{k-1}$. Let $P^{B}=P_{m}$ be the last blowup and
$\Sigma_{B}=\Sigma_{m}$ its fan.

To fully describe the fan, we will use the combinatorial approach to wonderful models developed in
\cite{Feichtner_2004}. Note that any fan $(\Sigma,\preceq)$ with its face relation is a meet
semi-lattice, i.e. every collection $\sigma_{1},\ldots,\sigma_{k}\in \Sigma$ has the greatest
lower bound
\begin{equation*}
  \bigwedge_{i} \sigma_{i} = \bigcap_{i} \sigma_{i}\in\Sigma.
\end{equation*}
The minimal element of $\Sigma$ is the trivial cone $\{0\}$. For $\Sigma=\Sigma_{E}$, there is an
obvious poset isomorphism
\begin{equation*}
  (\Sigma_{E},\preceq) \cong (2^{E}\backslash\{E\},\subseteq),
\end{equation*}
identifying a subset $I\subsetneq E$ with the cone $\tau_{I}$.

Now let $(\mathcal{L},\preceq)$ be any finite meet-semilattice with least element $\hat 0$. For a
subset $\mathcal{G}\subseteq\mathcal{L}$, and $X\in\mathcal{L}$, let
\begin{equation*}
  \mathcal{G}^{\preceq X} := \{G\in\mathcal{G}|G\preceq X\}
\end{equation*}
and
\begin{equation*}
  [\hat 0,X] = \{Y\in\mathcal{L}|0\preceq Y\preceq X\}.
\end{equation*}

\begin{defin}
  Let $(\mathcal{L},\preceq)$ be a finite meet-semilattice. A subset
  $\mathcal{G}\subseteq\mathcal{L}\backslash\{\hat 0\}$ is called a \emph{building set}, if the
  following holds for all $X\in\mathcal{L}$: Let
  \begin{equation*}
    \max \mathcal{G}^{\preceq X}=\{G_{1},\ldots,G_{k}\}
  \end{equation*}
  be the maximal elements of $\mathcal{G}^{\preceq X}$. Then there is an order isomorphism
  \begin{equation*} [\hat 0, X] \cong \prod_{i=1}^{k}[\hat 0, G_{i}].
  \end{equation*}
\end{defin}

\begin{exam}\label{exam:building-set}
  Suppose $\mathcal{L}=2^{E}\backslash \{E\}$ with the subset relation. Then a subset
  $\mathcal{G}\subseteq \mathcal{L}\backslash\emptyset$ is a building set if for all $I\subsetneq E$, the
  maximal elements
  \begin{equation*}
    \max \mathcal{G}^{\subseteq I}=\{G_{1},\ldots,G_{k}\}
  \end{equation*}
  form a partition $I=\coprod_{i=1}^{k}G_{i}$. It is easy to check that this is equivalent to the
  condition that $\mathcal{G}$ contains all singleton subsets and for all $I_{1},I_{2}\in
  \mathcal{G}$:
  \begin{equation*}
    I_{1}\cap I_{2}\neq \emptyset \Rightarrow I_{1}\cup I_{2}\in \mathcal{G} \text{ or } I_{1}\cup I_{2} = E.
  \end{equation*}
  The set $\mathcal{\tilde G} = \mathcal{G}\cup \{E\}$ is then a building set in $2^{E}$.
\end{exam}

To describe the face structure of $\Sigma_{B}$, we will also need the notion of nested sets.
\begin{defin}
  Let $\mathcal{G}\subseteq \mathcal{L}\backslash \hat 0$ be a building set in a finite
  meet-semilattice. A subset $\mathcal{N}\subseteq \mathcal{G}$ is called nested if for all pairwise
  non-comparable elements $N_{1},\ldots,N_{k}\in \mathcal{N}$ with $k\ge 2$, the join
  $\bigvee_{i=1}^{k}N_{i}\in\mathcal{L}$ exists in $\mathcal{L}$ but is not in $\mathcal{G}$.
\end{defin}

\begin{exam}\label{exam:nested-sets}
  Suppose $\mathcal{G}\subset 2^{E}\backslash\{E\}$ is a building set. Then
  $\mathcal{I}\subset \mathcal{G}$ is a nested set if and only if:
  \begin{enumerate}
  \item For all $I_{1},I_{2}$, either $I_{1}\cap I_{2}=\emptyset$ or $I_{1}\subseteq I_{2}$ or
    $I_{2}\subseteq I_{1}$.
  \item If $I_{1},\ldots,I_{k}\in \mathcal{I}$ are pairwise disjoint and $k\ge 2$, then
    \begin{equation*}
      \bigcup_{j=1}^{k} I_{j} \notin \mathcal{G} \cup \{E\}.
    \end{equation*}
  \end{enumerate}
  It follows from \cite[Prop. 2.8]{Feichtner_2004} that all maximal nested sets are generated by the
  following construction: Let $E=\{i_{1}<\ldots<i_{n}\}$ be a total ordering of $E$. Set
  $J_{k}=\{i_{1},\ldots,i_{k}\}$ and $\mathcal{I}_{k}=\max \mathcal{G}^{\subseteq J_{k}}$. The union
  $\mathcal{I}=\bigcup_{k=1}^{n} \mathcal{I}_{k}$ is then a maximal nested set.
\end{exam}

The nested sets of a building set $\mathcal{G}$ are partially ordered by inclusion. We denote the
corresponding poset by $\mathcal{N}(G)$. We can now state the results of \cite[Theorem
4.10]{Feichtner_2004}.
\begin{thm}
  Let $\Sigma$ be the fan of a toric variety $X_{\Sigma}$ and $\mathcal{G}\subseteq
  (\Sigma,\preceq)$ be a building set in its face semilattice. Suppose
  $\mathcal{G}=\{\tau_{1}<\ldots<\tau_{k}\}$ is linearly ordered, such that $\tau_{i}<\tau_{j}$
  implies $\tau_{j}\preceq \tau_{i}$. Let $\Sigma_{\mathcal{G}}$ be the fan obtained by subdividing
  $\Sigma$ along the $G_{i}$ in increasing order. Then there is an isomorphism of semilattices
  \begin{equation*}
    (\Sigma_{\mathcal{G}},\preceq) \cong (\mathcal{N}(G),\subseteq).
  \end{equation*}
  identifying a nested set $\mathcal{N}=\{\tau_{i_{1}},\ldots,\tau_{i_{s}}\}\subset G$ with the cone
  \begin{equation*}
    \tau_{\mathcal{N}} = \pos(\nu_{\tau}\ \rvert\ \tau\in N),
  \end{equation*}
  where $\nu_{\tau}=\sum_{\rho\in\tau(1)}u_{\rho}$.
\end{thm}

To our original set $B\subset 2^{E}$, we associate the set
\begin{equation*}
  \mathcal{G}_{B}=B\cup \{\{i\}\ \rvert\ i\in E \}.
\end{equation*}
This is a building set by example \ref{exam:building-set}. For $I\subseteq E$, we associate the
vector
\begin{equation*}
  e^{I} = \sum_{i\in I}e^{i}.
\end{equation*}
Applying the above theorem to $\mathcal{G}_{B}$ then gives:

\begin{col}\label{col:blowup-nested-sets}
  $P^{B}$ is a smooth, projective variety, independent of the chosen blowup-order. Its fan
  $\Sigma_{B}$ consists of the cones
  \begin{equation*}
    \sigma_{\mathcal{I}} = \pos([e^{I}]\ \rvert\ I\in \mathcal{I}),
  \end{equation*}
  where $\mathcal{I}\subset \mathcal{G}_{B}$ ranges over the nested sets with respect to
  $\mathcal{G}_{B}$.
\end{col}


In particular we have $\Sigma_{B}(1)\cong B\cup E$. The map $P^{B}\rightarrow P^{E_{G}}$ fits into the
commutative diagram
\begin{center}\begin{tikzcd}
    \mathbb{C}^{\Sigma_{B}(1)}\backslash Z_{\Sigma_{B}}\arrow{d}\arrow{r} & P^B \arrow{d} \\
    \mathbb{C}^{E}\backslash \{0\} \arrow{r} & P^E
  \end{tikzcd}\end{center}
The left vertical map is given on coordinates as
\begin{equation*}
  \alpha_{i} = x_{i}\prod_{\substack{I\in B\\ i\in I}}x_{I}.
\end{equation*}

\paragraph{Generalized permutahedra.}
\label{sec:gener-perm}
The Feynman polytopes we consider later will turn out to be generalized permutahedra in the sense of
(\cite{Postnikov_2009},\cite{aguiar17:hopf}), which have an especially nice structure.

Consider first the building set $G_{max}=2^{E}\backslash\{E\}$. Its fan $\Sigma_{G_{max}}$ is spanned
by cones $\sigma=\pos(e_{I_{1}},\ldots, e_{I_{n-1}})$ such that
\begin{equation*}
  I_{0}=\emptyset\subsetneq I_{1}\subsetneq \ldots\subsetneq I_{n-1}\subsetneq I_{n}=E
\end{equation*}
is a complete flag of subsets. On the other hand, let $\pi_{E}$ be the convex hull of all points
\begin{equation*}
  a_{\sigma} = \sum_{k=1}^{n}ka_{\sigma(k)},
\end{equation*}
where $\sigma$ runs over the bijections $\{1,\ldots,n\}\cong E$. The polytope $\pi_{E}$ is the
(regular) permutahedron of the finite set $E$. The following proposition is then well-known, see
e.g. \cite{Postnikov_2009}.

\begin{prop}
  The normal fan $\Sigma_{\pi_{E}}$ of $\pi_{E}$ coincides with $\Sigma_{G_{max}}$.
\end{prop}

The facet structure of $\pi_{E}$ is very well understood. It is advocated in \cite{aguiar17:hopf} to
exploit this fact by expressing many questions in algebraic combinatorics in terms of deformations
of $\pi_{E}$:

\begin{defin}(\cite{aguiar17:hopf}) A lattice polytope $P\subset \mathbb{R}^{E}$ contained in an
  affine hyperplane
  \begin{equation*}
    P \subset \{m\in \mathbb{R}^{E}\ \rvert\ \langle m, e^{E}  \rangle = d_{P}\}
  \end{equation*}
  is a a \emph{generalized permutahedron} if its normal fan is a coarsening of the fan
  $\Sigma_{\pi_{E}}$.
\end{defin}

It will be convenient to have alternative characterizations of generalized permutahedra. Suppose
$z:2^{E}\rightarrow \mathbb{Z}\cup \{\infty\}$ is a set function with $z(\emptyset) = 0$. To $z$ we
associate the base polyhedron
\begin{equation*}
  P(z) = \{m\in \mathbb{R}^{E} \ \rvert\
  \langle m, e^{E}  \rangle = z(E),\ \langle m, e^{I} \rangle \ge z(I) \text{ for } I\subsetneq E \}.
\end{equation*}
We will call $z$ \emph{supermodular}, if
\begin{equation*}
  z(I) + z(J) \le z(I\cap J) + z(I\cup J),
\end{equation*}
for all $I,J\in 2^{E}$.

\begin{remark}
  It is more common in the literature to consider \emph{submodular} functions $\tilde
  z:2^{E}\rightarrow \mathbb{R}\cup \{\infty\}$, which satisfy the opposite inequality:
  \begin{equation*}
    \tilde z(I) + \tilde z(J) \ge \tilde z(I\cap J) + \tilde z(I\cup J)
  \end{equation*}

  It is easy to show that $\tilde z$ is submodular if and only if its dual $\tilde z^{\#}$, defined
  as $\tilde z^{\#}(I)=\tilde z(E)-\tilde z(E \backslash I)$, is supermodular. The translation
  between the two convention is usually straightforward.
\end{remark}

\begin{prop}
  Let $P\subset \mathbb{R}^{E}$ be a lattice polytope. Then the following are equivalent:
  \begin{enumerate}
  \item $P$ is a generalized polyhedron.
  \item Every edge of $P$ is parallel to an edge of the form $e_{i}-e_{j}$ for $i,j\in E$.
  \item There is a supermodular function $z:2^{E}\rightarrow \mathbb{R}$, such that $P=P(z)$.
  \end{enumerate}
\end{prop}
\begin{proof}
  See \cite[Thm. 12.3]{aguiar17:hopf} and references therein.
\end{proof}

\begin{exam}\label{exam:matroid}
  Let $M$ be a matroid on the set $E$ and $B(M)\subset 2^{E}$ its set of bases. We refer to
  \cite{oxley2006matroid} for the theory of matroids. The matroid polytope of $M$ is
  \begin{equation*}
    P_{M} = \Conv(e^{I}\ \rvert\ I\in B(M)).
  \end{equation*}
  It is proven in \cite{Gelfand_1987} that every edge of $P_{M}$ is of the form $e^{i}-e^{j}$, hence
  $P_{M}$ is a generalized permutahedron. The corresponding supermodular function is given by
  \begin{equation*}
    z(I) = r_{M}(E)-r_{M}(E \backslash I) = r^{\#}(I),
  \end{equation*}
  where $r_{M}$ is the rank function of the matroid.
\end{exam}

\begin{exam}
  Let $\mathcal{G}\subseteq 2^{E}\backslash\{E\}$ be a building set and $\mathcal{\tilde
    G}=\mathcal{G}\cup \{E\}$. For $I\in \mathcal{\tilde G}$, let $\Delta_{I}=\Conv(e^{i} \ \rvert\ i\in
  I)$ be the simplex on $I$ and let $P_{\mathcal{G}}$ be the Minkowski sum
  \begin{equation*}
    P_{\mathcal{G}} = \sum_{I\in \mathcal{\tilde G}}\Delta_{I}.
  \end{equation*}
  It is shown in \cite{math/0609184}, that its normal fan is $\Sigma_{\mathcal{G}}$ and that
  $P_{\mathcal{G}}$ is the base polyhedron of the supermodular function
  \begin{equation*}
    z_{\mathcal{G}}(J) = |\{I\in \mathcal{\tilde G} \ \rvert\ I\subseteq J\}|.
  \end{equation*}
  Hence $P_{\mathcal{G}}$ is a generalized permutahedron and $\Sigma_{\mathcal{G}}$ is a coarsening
  of $\Sigma_{\pi_{E}}$.

  Suppose $\mathcal{G}_{1}\subset \mathcal{G}_{2}$ are two building sets in $2^{E}\backslash\{E\}$.
  It follows from the above description and Prop. \ref{prop:common-refinement} that
  $\Sigma_{\mathcal{G}_{2}}$ is a refinement of $\Sigma_{\mathcal{G}_{1}}$.
\end{exam}

For $I\subsetneq E$ define the restriction $z\vert_{I}$ and contraction $z/_{I}$ by
\begin{align*}
  z\vert_{I}(J) &= z(J),\quad J\subseteq I,\\
  z/_{I}(J) &= z(J\cup I)-z(I),\quad J\subseteq E \backslash I
\end{align*}

It is easy to check that if $z$ is supermodular, then so are its restrictions and contraction. The
face $F_{e^{I}}P(z)$ can then be described as follows.
\begin{prop}[{\cite[Lemma 3.1]{MR1095782}}]\label{prop:supermodular-faces}
  Let $P(z)$ be the generalized permutahedron defined by the supermodular function
  $z:2^{E}\rightarrow \mathbb{R}$. The natural isomorphism $\mathbb{R}^{I}\oplus
  \mathbb{R}^{I^{c}}\cong \mathbb{R}^{E}$ induces a bijection
  \begin{equation*}
    P(z\vert_{I})\times P(z/_{I}) \cong F_{e^{I}}P(z).
  \end{equation*}
\end{prop}

\begin{exam}\label{exam:matroid-restr-contract}
  If $z=r^{\#}$ is the dual of the rank function of a matroid $M$ on $E$, then $z\vert_{I}$ and
  $z/_{I}$ correspond to the contraction $M/_{I^{c}}$ and restriction $M\vert_{I^{c}}$.
\end{exam}

Let
\begin{equation*}
  \mathcal{I}: I_{0}=\emptyset\subsetneq I_{1}\subsetneq \ldots\subsetneq I_{n-1}\subsetneq I_{n}=E
\end{equation*}
be a maximal flag of $2^{E}$. The corresponding cone $\sigma_{\mathcal{I}}=\pos( e^{I} \ \rvert\
I\in\mathcal{I})$ is a maximal cone of $\Sigma_{\pi_{E}}$. Since $\Sigma_{\pi_{E}}$ is a refinement
of $\Sigma_{P(z)}$, any vector $w\in \Int(\sigma)$ defines a vertex $m_{\mathcal{I}}=F_{w}P(z)$.

\begin{prop}{\cite[Corollary 3.17]{MR1095782}}\label{prop:supermodular-vertex-generation}
  The coordinates of the vertex $m_{\mathcal{I}}$ are given by
  \begin{equation*}
    (m_{\mathcal{I}})_{k}=z(I_{k})-z(I_{k-1}).
  \end{equation*}
\end{prop}

Let us call a generalized permutahedron $P(z)$ \emph{irreducible}, if there is no decomposition $E =
I \coprod J$, such that $z = z\vert_{I}+z\vert_{J}$.

\begin{prop}{\cite[Thm. 3.38]{MR1095782}}\label{prop:supermodular-decomposition}
  For each generalized permutahedron $P(z)$ there is a unique decomposition
  $E=\coprod_{k=1}^{r}I_{k}$ such that the $P(z\vert_{I_{k}})$ are irreducible and
  \begin{equation*}
    z = \sum_{k=1}^{r}z\vert_{I_{k}}.
  \end{equation*}
  The polytope $P(z)$ is irreducible if and only if it has maximal dimension $|E|-1$.
\end{prop}

\begin{col}\label{col:supermodular-facets}
  Suppose $P(z)$ is irreducible. A subset $I\subsetneq E$ defines a facet $F_{e^{I}}P(z)$ of $P(z)$
  if and only if $P(z\vert_{I})$ and $P(z/_{I})$ are both irreducible.
\end{col}

Let us use the preceding results to construct a smooth refinement of $P(z)$. Consider the subset
system
\begin{equation*}
  \mathcal{\tilde G}_{z} = \{I\subseteq E \ \rvert\ P(z\vert_{I}) \text{ is irreducible }\} \subseteq 2^{E}
\end{equation*}

\begin{prop}\label{prop:supermodular-wonderful-refinement}
  $\mathcal{\tilde G}_{z}$ is a building set in $2^{E}$. If $P(z)$ is irreducible, then the fan $\Sigma_{\mathcal{G}_{z}}$
  associated to the reduced building set $\mathcal{G}_{z}=\mathcal{G}
  \backslash\{E\}$ is a smooth refinement of $\Sigma_{P(z)}$.
\end{prop}

\begin{proof}
  The building set property is immediate from the unique decomposition of Prop.
  \ref{prop:supermodular-decomposition}.

  To prove that $\Sigma_{\mathcal{G}_{z}}$ is a refinement of $\Sigma_{P(z)}$, we must prove that
  for each maximal nested set $\mathcal{I}\subset \mathcal{G}_{z}$, there is $m\in P(z)$ such that
  $\langle m, e^{I} \rangle = z(I)$ for all $I\in\mathcal{I}$. By Example \ref{exam:nested-sets} we
  can find a maximal chain
  \begin{equation*}
    J_{0}=\emptyset\subsetneq J_{1}\subsetneq \ldots\subsetneq J_{n-1}\subsetneq J_{n}=E
  \end{equation*}
  such that $\mathcal{I}=\bigcup \mathcal{I}_{k}$, where $\mathcal{I}_{k}=\max
  \mathcal{G}_{z}^{\subseteq J_{k}}$. Let $m$ be the vertex of $P(z)$, defined by
  $m_{k}=z(J_{k})-z(J_{k-1})$. Since $\mathcal{I}_{k}$ is the decomposition of $J_{k}$ into
  irreducible components, we have
  \begin{equation*}
    \sum_{I\in\mathcal{I}_{k}} z(I) = z(J_{k}) = \langle m, e^{J_{k}} \rangle = \sum_{I\in\mathcal{I}_{k}} \langle m, e^{I} \rangle.
  \end{equation*}
  Since $m\in P(z)$, this equality is only possible if $\langle m, e^{I} \rangle=z(I)$ for all
  $I\in\mathcal{I}_{k}$.
\end{proof}

\section{Multivariate Mellin transforms}
\label{sec:analyt-dimens-regul}

Let us apply the theory of toric varieties to the investigation of Mellin transforms of Laurent
polynomials. It was shown in (\cite{Nilsson_2011}, \cite{Berkesch_2014}), that the convergence
properties of these transforms are controlled by the Newton polytopes of the rational functions. We
supply an alternative proof of their results by considering certain toric compactification
associated to the Newton polytopes, which make the possible singularities apparent.
This gives a precise characterisation of the convergence domain.

For application to dimensional regularization, we will also review their construction of the
meromorphic extension.
In the last part of this section, we will show that the geometric sector decomposition strategy of
Kaneko and Ueda (\cite{Kaneko_2010}) is equivalent to the construction of these compactifications.

\paragraph{Mellin transforms.}
\label{sec:gener-mell-transf}

As in the previous section, let $N$ be a lattice of finite rank, $M$ its dual lattice and $T_{N}$
the associated complex torus. Let
\begin{equation*}
  f(z)=\sum_{m\in A}a_{m}t^{m}\in \mathcal{O}(T_{N})
\end{equation*}
be a Laurent polynomial on $T_{N}$, where $A\subset \mathbb{Z}^{n}$ is a finite subset such that
$a_{m}\neq 0$ for $m\in A$. Its \emph{Newton} polytope is the convex hull
\begin{equation*}
  P(f) = \Conv(A)\subset M_{\mathbb{R}}.
\end{equation*}
Note that $P(f\cdot g)=P(f)+P(g)$.

Suppose $f_{1},\ldots,f_{k}$ are a collection of Laurent polynomials as above. Assume that all
non-vanishing coefficients of the $f_{i}$ are contained in an open, strongly convex polyhedral cone
$U\subset \mathbb{C}^{N_{i}}$. Then the complex powers $f_{i}(t)^{c_{i}}$ are well-defined for every
$c_{i}\in \mathbb{C}$, $t\in T_{N}(\mathbb{R}^{+})$ and fixed choice of branch of $w\mapsto
w^{c_{i}}=e^{c_{i}\log(w)}$. These conditions imply that each $f_{i}$ is totally non-vanishing on
$T_{N}(\mathbb{R}^{+})$ in the sense of \cite{Nilsson_2011}.

Choose a $\mathbb{Z}$-basis of $N$ such that $N\cong \mathbb{Z}^{n}$ and inducing isomorphisms
$M\cong \mathbb{Z}^{n}, T_{N}\cong (\mathbb{C}^{*})^{n}$. Let
\begin{equation*}
  \Log: T_{N} \rightarrow N_{\mathbb{C}}, \quad t\mapsto (\log(t_{1}),\ldots,\log(t_{n}))
\end{equation*}
be the componentwise logarithm. For $s\in M_{\mathbb{C}}\cong \mathbb{C}^{n}$, we define the
multivalued monomial
\begin{align*}
  t^{s} =  \prod_{i=1}^{n}t_{i}^{s_{i}} =  e^{\langle s, \Log t \rangle},
\end{align*}
These definitions are clearly independent of the choice of basis. Let us also denote by
\begin{equation*}
  \frac{\df t}{t} = \frac{\df t_{1}}{t_{1}} \wedge \ldots \wedge \frac{\df t_{n}}{t_{n}}
\end{equation*}
the holomorphic $T_{N}$-invariant volume form, which is independent of the choices up to sign, i.e.
up to the choice of an orientation of $N$.

We are interested in the analytic properties of the multivariate Mellin transform
\begin{equation*}
  \mathcal{M}(f_{i},s,c)=\int_{T_{N}(\mathbb{R^{+}})}t^{s}\prod_{i=1}^{k} f_{i}(t)^{-c_{i}}
  \frac{\df t}{t}.
\end{equation*}

We will see that the convergence properties of the above integral are governed by the polytope
\begin{equation*}
  P=P(f_{1})+\ldots+P(f_{k})=P(f_{1}\cdots f_{k}).
\end{equation*}
For a weight vector $u\in N_{\mathbb{R}}$ and Laurent polynomial $h=\sum_{m}h_{m}t^{m}$, let
\begin{equation*}
  d_{u}(h) = \min_{m\in P(h)}\langle m, u \rangle.
\end{equation*}
If $u=u_{\rho}\in \mathbb{Z}$ is the lattice generator of a rational ray $\rho\subset
N_{\mathbb{R}}$, then we define $d_{\rho}(h)=d_{u_{\rho}}(h)$. For $c\in \mathbb{C}^{k}$ let us also
set
\begin{equation*}
  d_{u}(c) = \sum_{i=1}^{k}c_{i}d_{u}(f_{i})
\end{equation*}
and $d_{\rho}(c)=d_{u_{\rho}}(c)$.

Let $\Sigma_{P}$ be the normal fan of the Newton polytope $P=P(f_{1}\cdots f_{n})$.
and denote by $\Lambda(f_{i})\subseteq M_{\mathbb{C}}\times
  \mathbb{C}^{k}$ the region of pairs $(s,c)\in M_{\mathbb{C}}\times \mathbb{C}^{k}$ satisfying
  \begin{equation*}
    \Real\langle s, u_{\rho} \rangle > \Real d_{\rho}(c)
  \end{equation*}
  for all $\rho\in\Sigma_{P}(1)$.

\begin{thm}\label{thm:mellin-convergence}
  If the polytope $P$ is full-dimensional then $\Lambda(f_{i})$
  is nonempty and the Mellin
  transform converges for $(s,c)\in M_{\mathbb{C}}\times \mathbb{C}^{k}$ if and only if $(s,c)\in\Lambda(f_{i})$.

  Conversely, if $P$ is not full-dimensional then the Mellin transform is not absolutely convergent
  for any choice of $(s,c)\in M_{\mathbb{C}}\times \mathbb{C}^{k}$.
\end{thm}

\begin{remark}
  For $r=(r_{1},\ldots,r_{k})\in (0,\infty)^{k}$, let
  \begin{equation*}
    P(r) = r_{1}P_{f_{1}}+\ldots+r_{k}P_{f_{n}}\subset M_{\mathbb{R}}
  \end{equation*}
  be the Minkowski sum of the scaled Newton polytopes. This polytope is full-dimensional if
  $P=P(1,\ldots,1)$ is full-dimensional. Then the interior of $P(r)$ is non-empty and $P(r)$ has the
  facet presentation
  \begin{equation*}
    P(r) = \bigcap_{\rho\in\Sigma_{P}(1)}\{\langle m, u_{\rho} \rangle\ge d_{\rho}(r)\}.
  \end{equation*}
  Hence the region $\Lambda(f_{i})$ contains the set
  \begin{equation*}
    \{(s,c)\in M_{\mathbb{C}}\times \mathbb{C}^{k}\ \rvert\ \Real(c)\in (0,\infty)^{k}, \Real(s)\in \Int(P(\Real(c)))\}.
  \end{equation*}
  This recovers the corresponding results of (\cite{Nilsson_2011},\cite{Berkesch_2014}).
\end{remark}

The basic idea of the proof is to find a compactification $X_{\Sigma}$ of $T_{N}$ such that the
strict transforms $\overline{V(f_{i})}$ do not intersect the real locus
$X_{\Sigma}(\mathbb{R}^{+})$. The convergence of the integral then essentially reduces to the
absence of poles along the divisor $D_{\Sigma}=X_{\Sigma}\backslash T_{N}$.

\begin{prop}\label{prop:simplicial-refinement}
  Suppose $f=\sum_{m\in A}a_{m}t^{m}$ is a Laurent polynomial on $T_{N}$, such that the coefficients
  $a_{m}$ are contained in a strongly convex polyhedral cone $U\subset \mathbb{C}^{A}$. Let
  $X_{\Sigma}$ be a complete, simplicial toric variety with torus $T_{N}$. Then the following are
  equivalent:
  \begin{enumerate}
  \item The fan $\Sigma$ is a refinement of the (possibly degenerate) normal fan $\Sigma_{P(f)}$ of
    the Newton polytope of $f$.
  \item The closure of the zero set $\overline{ V(f) }$ of $f$ does not intersect the real positive
    locus $X_{\Sigma}(\mathbb{R}^{+})$.
  \end{enumerate}
\end{prop}

\begin{proof}[Proof of Prop. \ref{prop:simplicial-refinement}]
  The variety $X_{\Sigma}$ is covered by the orbifold charts
  $U_{\sigma}=\mathbb{C}^{\sigma(1)}//G_{\sigma}$, where
  $\sigma=\pos(u_{1},\ldots,u_{n})\in\Sigma(n)$ is a maximal cone. The Laurent monomials are
  expressed in the coordinates $x_{i}$ of $U_{\sigma}$ as
  \begin{equation*}
    t^{m} = \prod_{i=1}^{n}x_{i}^{\langle m,u_{i}  \rangle} =: x^{m}.
  \end{equation*}
  
  Suppose $\sigma$ is contained in a maximal cone of $\Sigma_{P}$. By Prop. \ref{prop:maximal-cones},
  there is $m_{\sigma}\in P(f)$ with
  \begin{equation*}
    \langle m_{\sigma}, u_{i} \rangle = \min_{m\in P(f)}\langle m, u_{i} \rangle,
  \end{equation*}
  for all $i=1,\ldots,n$. The Laurent polynomial $f$ is then expressed in these coordinates as
  \begin{align*}
    f(x) &= \sum_{m\in A}a_{m}x^{m}  = x^{m_{\sigma}}\left(a_{m_{\sigma}}+
           \sum_{m\in A \backslash\{m_{\sigma}\}}a_{m}x^{m_{\sigma}} \right) \\
         &=: x^{m_{\sigma}}f_{\sigma}(x)
  \end{align*}
  The polynomial $f_{\sigma}(x)$ is regular and non-vanishing on $U_{\sigma}(\mathbb{R}^{+})$. It
  follows that
  \begin{equation*}
    \overline{V(f)}\cap U_{\sigma}(\mathbb{R}^{+}) = V(f_{\sigma})\cap U_{\sigma}(\mathbb{R}^{+})
    =\emptyset.
  \end{equation*}
  Conversely, suppose $\overline{V(f)}\cap U_{\sigma}(\mathbb{R}^{+})$ is empty. The Zariski closure
  $\overline{V(f_{i})}\subset U_{\sigma}$ is described by a polynomial $\tilde f\in
  k[x_{1},\ldots,x_{n}]$ such that $f(x)=x^{\tilde m}\tilde f(x)$, i.e. $\tilde f$ has the form
  \begin{equation*}
    \tilde f(x) = \sum_{m\in A}a_{m}x^{m-\tilde m}.
  \end{equation*}
  Since $\tilde f$ is regular on $\mathbb{C}^{n}$, we must have $\langle m-\tilde m, u_{i}
  \rangle\ge 0$ for all $m\in A$ and $i=1,\ldots,k$. The intersection $V(\tilde f)\cap
  U_{\sigma}(\mathbb{R}^{+})$ can only be empty if there is at least one $m\in A$, such that
  $\langle m-\tilde m, u_{i} \rangle=0$ for all $i$. But this would imply that $\langle m, u_{i}
  \rangle$ is minimal for all $i=1,\ldots,n$ and Prop. \ref{prop:maximal-cones} shows that $\sigma$
  is contained in a maximal cone of $\Sigma_{P(f)}$.

  Since the open sets $U_{\sigma}$ cover $X_{\Sigma}$, we have shown that
  \begin{equation*}
    \overline{V(f)}\cap
    X_{\Sigma}(\mathbb{R}^{+})=\emptyset
  \end{equation*}
  if and only if every maximal cone of $\Sigma$ is contained in a cone of $\Sigma_{P}$, which means
  that $\Sigma$ is a refinement of $\Sigma_{P}$.
\end{proof}

By Prop. \ref{prop:common-refinement}, the normal fan $\Sigma_{P}$ of $P$ is the coarsest common
refinement of the normal fans of the Newton polytopes $P(f_{i})$. Suppose $\Sigma$ is a simplicial
refinement of $\Sigma_{P}$. The above proposition shows that $\overline{V(f_{i})}\cap
X_{\Sigma}(\mathbb{R}^{+})=\emptyset$ for all $i$. By (a straightforward generalization of) Prop.
\ref{prop:differential-form} we can express the Mellin transform as
\begin{equation*}
  \mathcal{M}(f_{i},s,c) = \int_{X_{\Sigma}(\mathbb{R}^{+})}
  \prod_{\rho\in\Sigma(1)}x_{\rho}^{\langle s, u_{\rho} \rangle -1}\prod_{i=1}^{k}F_{i}(x)^{-c_{i}} \Omega_{\Sigma},
\end{equation*}
where $F_{i}(x)$ is the homogenization of $f_{i}$.

\begin{remark}
  The above integral should be understood in the orbifold sense. More precisely, let
  $U_{\sigma}=\mathbb{C}^{n}//G_{\sigma}$ be an orbifold chart and $\alpha$ a compactly supported,
  $G_{\sigma}$-invariant $n$-form on $\mathbb{C}^{n}$, absolutely integrable on $(\mathbb{R}^{+})^{n}$. We
  define the integral over $U_{\sigma}(\mathbb{R}^{+})$ as
  \begin{equation*}
    \int_{U_{\sigma}(\mathbb{R}^{+})}\alpha = \frac{1}{|G_{\sigma}|}\int_{(\mathbb{R}^{+})^{n}}\alpha.
  \end{equation*}
  For a general analytic section $\alpha$ of $\omega_{\Sigma}$ which is locally integrable on
  $X_{\Sigma}(\mathbb{R}^{+})$, we define
  \begin{equation*}
    \int_{X_{\Sigma}(\mathbb{R}^{+})}\alpha = \sum_{\sigma\in\Sigma(n)}\int_{U_{\sigma}(\mathbb{R}^{+})}\rho_{\sigma}\alpha,
  \end{equation*}
  where $(\rho_{\sigma})$ is a smooth partition of unity subordinate to the cover
  $(U_{\sigma})_{\sigma\in\Sigma(n)}$. See \cite[Section 2.1]{Adem_2007} for further details.
\end{remark}

\begin{prop}
  Let $\Sigma$ be a simplicial refinement of $\Sigma_{P}$. Then the above integral converges
  absolutely if and only if
  \begin{equation*}
    \Real\langle s, u_{\rho} \rangle > \Real d_{\rho}(c),
  \end{equation*}
  for all $\rho\in\Sigma(1)$.
\end{prop}

\begin{proof}
  The fan $\Sigma$ is complete since it is a refinement of the normal fan of a lattice polytope.
  Hence the integration domain is compact and it is enough to show that the integrand is locally
  integrable. Let
  \begin{equation*}
    \sigma=\pos(u_{1},\ldots,u_{n})\in\Sigma(n)
  \end{equation*}
  be a maximal cone. The proof of the previous proposition shows that
  \begin{equation*}
    F_{i}|_{U_{\sigma}} = \prod_{j=1}^{n}x_{j}^{\langle m_{\sigma,i}, u_{j} \rangle}f_{\sigma,i},
  \end{equation*}
  where $f_{\sigma,i}$ does not vanish on $U_{\sigma}(\mathbb{R}^{+})$ and $m_{\sigma,i}$ satisfies
  \begin{equation*}
    \langle m_{\sigma,i}, u_{j} \rangle = d_{u_{j}}(f_{i}).
  \end{equation*}
  The integrand then has the local expression
  \begin{align*}
    \prod_{j=1}^{n}\prod_{i=1}^{k}x_{j}^{\langle s, u_{j} \rangle -1}F_{i}(x)^{-c_{i}}\big\vert_{U_{\sigma}}
    &= \prod_{j=1}^{n}x_{j}^{\langle s, u_{j} \rangle-d_{\rho}(c) -1}\prod_{i=1}^{k}f_{\sigma,i}(x)^{-c_{i}}.
  \end{align*}
  This is locally integrable on $U(\mathbb{R}^{+})\cong \mathbb{R}^{n}_{+}//G_{\sigma}$ if and only
  if
  \begin{equation*}
    \Real(\langle s, u_{j} \rangle -d_{u_{j}}(c)) > 0.
  \end{equation*}
  for all $\mathbb{R}^{+}u_{j}\in\sigma(1)$. Letting $\sigma\in\Sigma(n)$ vary over all
  $n$-dimensional cones gives the result.
\end{proof}

\begin{proof}[Proof of Thm. \ref{thm:mellin-convergence}]
  Suppose $P$ is full-dimensional. From Theorem \ref{thm:simplicial-refinement}, we can construct a
  simplicial refinement $\Sigma$ of $\Sigma_{P}$ with $\Sigma(1)=\Sigma_{P}(1)$. The previous
  proposition shows that the integral is convergent if $(s,c)$ satisfy
  \begin{equation*}
    \Real\langle s, u_{\rho} \rangle > \Real d_{\rho}(c),
  \end{equation*}
  for all $\rho\in\Sigma(1)=\Sigma_{P}(1)$.

  Now suppose $P(r)$ is not full-dimensional, i.e. it is contained in a hyperplane
  \begin{equation*}
    P \subseteq \{m\in M_{\mathbb{R}}\ \rvert\  \langle m, u \rangle = d\},
  \end{equation*}
  This is only possible, if $P(f_{i})$ is contained in the hyperplane
  \begin{equation*}
    \{m\in M_{\mathbb{R}} \ \rvert\ \langle m, u\rangle=d_{u}(f_{i})\}.
  \end{equation*}
  Every simplicial refinement $\Sigma$ of $\Sigma_{P}$ must contain the cones
  $\rho^{\pm}=\mathbb{R}^{\pm}u$. The inequalities
  \begin{equation*}
    \Real\langle s, u \rangle > \Real d_{u}(c) \quad \Real\langle s,-u \rangle > \Real d_{-u}(c).
  \end{equation*}
  corresponding to $\rho^{+}$ and $\rho^{-}$ are not both satisfiable, since
  \begin{equation*}
    d_{u}(c) = \sum_{i=1}^{k}c_{i}d_{u}(f_{i}) = -d_{-u}(c).
  \end{equation*}
  Thus the integral does not convergence for any choice of $(s,c)\in M_{\mathbb{C}}\times
  \mathbb{C}^{k}$.
\end{proof}

\paragraph{Analytic continuations.}

Suppose $g=\sum_{m\in B}b_{m}t^{m}$ is another Laurent polynomial with Newton polytope $P(g)$.
Theorem \ref{thm:mellin-convergence} implies that the integral
\begin{equation*}
  \mathcal{M}(f_{i},g,s,c) = \int_{T_{N}(\mathbb{R^{+}})}t^{s}g(t)\prod_{i=1}^{k}f_{i}(t)^{-c_{i}} \frac{\df t}{t} 
\end{equation*}
is convergent if
\begin{equation*}
  \Real \langle s+m, u_{\rho} \rangle > \Real(d_{\rho}(c)),
\end{equation*}
for all $\rho\in\Sigma_{P}(1)$ and $m\in B$. This is clearly equivalent to
\begin{equation*}
  \Real \langle s, u_{\rho} \rangle > \Real(d_{\rho}(c))-d_{\rho}(g).
\end{equation*}
Let $\Lambda(f,g)\subset M_{\mathbb{C}}\times \mathbb{C}^{k}$ be the open set of parameters $(s,c)$
satisfying the above inequalities. The articles \cite{Nilsson_2011} and \cite{Berkesch_2014}
construct a meromorphic continuation of $\mathcal{M}(f_{i},g,s,c)$ to $M_{\mathbb{C}}\times
\mathbb{C}^{k}$. Let us briefly sketch their argument.

For a ray $\rho\in\Sigma_{P}(1)$ with lattice generator $u_{\rho}\in N$, let
\begin{equation*}
  \mathbb{C}^{*}\rightarrow T_{N},\quad \lambda \mapsto \lambda_{\rho}:= u_{\rho} \otimes \lambda
\end{equation*}
be the one-parameter subgroup defined by $u_{\rho}$. Composing with the left action of $T_{N}$ and
restricting to $(0,\infty)\subset \mathbb{C}^{*}$ gives the action
\begin{equation*}
  (0,\infty)\times T_{N} \rightarrow T_{N},\quad (\lambda,t)\mapsto \lambda_{\rho}\cdot t.
\end{equation*}

Let $h(t)=\sum_{m\in C}h_{m}t^{m}$ be an arbitrary Laurent polynomial. Set
\begin{equation*}
  h_{\rho}(t) = \frac{\df}{d\lambda}\left( \frac{h(\lambda_{\rho}\cdot t)}{\lambda^{d_{\rho}(h)}} \right){\bigg\rvert}_{\lambda=1}.
\end{equation*}
For monomials $m\in F_{u_{\rho}}P(h)$, $(\lambda_{\rho}\cdot t)^{m}=\lambda^{d_{\rho}(h)}t^{m}$,
which implies that
\begin{equation*}
  d_{\rho}(h_{\rho}) \ge d_{\rho}(h)+1.
\end{equation*}
The differential form $\frac{\df t}{t}$ and integration domain $T_{N}(\mathbb{R^{+}})$ are invariant
under the action of $\lambda_{\rho}$. Hence we have
\begin{align*}
  \mathcal{M}(f_{i},g,s,c) &= \int_{T_{N}(\mathbb{R^{+}})}(\lambda_{\rho}\cdot t)^{s}g(\lambda_{\rho}\cdot t)\prod_{i=1}^{k}f_{i}(\lambda_{\rho}\cdot t)^{-c_{i}} \frac{\df t}{t} \\
                           &= \int_{T_{N}(\mathbb{R^{+}})}\lambda^{\langle s, u_{\rho}\rangle -d_{ \rho}(c)+d_{\rho}(g)}
                             t^{s}\frac{g(\lambda_{\rho}\cdot t)}{\lambda^{d_{\rho(g)}}}
                             \prod_{i=1}^{k}\left(\frac{f_{i}(\lambda_{\rho}\cdot t)}
                             {\lambda^{d_{\rho,i}}}\right)^{-c_{i}} \frac{\df t}{t}.
\end{align*}
Differentation with respect to $\lambda$ and setting $\lambda=1$ gives
\begin{equation*}
  \mathcal{M}(f_{i},s,c,g) = \frac{-1}{\langle s, u \rangle - d_{\rho}(c)+d_{\rho}(g)}\left( I_{g}(s,c)-\sum_{i=1}^{k}c_{i}I_{i}(s,c) \right),
\end{equation*}
where
\begin{align*}
  I_{g}(s,c) &= \mathcal{M}(f_{i},g_{\rho},s,c) \\
  I_{i}(s,c) &= \mathcal{M}(f_{i},gf_{\rho},s,c+e_{i}).
\end{align*}
The integral $I_{g}(s,c)$ converges for $(s,c)$ satisfying
\begin{align*}
  \Real(\langle s, u_{\tilde \rho} \rangle) > \Real(d_{\tilde \rho}(c)) - d_{\tilde \rho}(g_{ \rho}),
\end{align*}
for all $\tilde\rho\in\Sigma_{P}(1)$. Since $P(g_{\rho})\subset P(g) $ and $d_{ \rho}(g_{ \rho})\ge
d_{ \rho}(g) +\delta^{\tilde \rho}_{\rho}$ by construction, we have
\begin{equation*}
  \Real(d_{\rho}(c)) - d_{\rho}(g_{ \rho})
  \le \Real(d_{\rho}(c)) - d_{\rho}(g) - \delta^{ \tilde \rho}_{\rho}.
\end{equation*}
Similarly, $I_{i}(s,c)$ converges iff
\begin{align*}
  \Real(\langle s, u_{\tilde \rho} \rangle) > \Real(d_{\tilde \rho}(c+e_{i})) - d_{\tilde \rho}(g) - d_{\tilde \rho}(f_{i, \rho}).
\end{align*}
From $d_{\tilde \rho}(f_{i, \rho})\ge d_{i,\tilde \rho}+\delta^{\tilde\rho}_{\rho}$ we get the
inequality
\begin{align*}
  \Real(d_{\tilde\rho}(c+e_{i})) - d_{\tilde\rho}(g) - d_{\tilde\rho}(f_{i, \rho}) 
  &\le \Real(d_{\tilde \rho}(c)) - d_{\tilde\rho}(g) - \delta^{\tilde \rho}_{\rho}.
\end{align*}
Hence the integrals $I_{g}(s,c),I_{i}(s,c)$ converge if
\begin{equation*}
  \langle \Real(s), u_{\tilde \rho} \rangle > \Real(d_{\tilde \rho}(c))- d_{\tilde \rho}(g) - \delta_{\rho}^{\tilde \rho}
\end{equation*}
for all $\tilde \rho\in\Sigma_{P}(1)$. Thus we have found an analytic continuation which improves
the convergence in the direction of $\rho$.

Iteratively differentiating with respect to one-parameter subgroups $\varphi_{\rho}(\lambda)$ as above
then gives the following result.
\begin{thm}[{\cite[Theorem 2.4]{Berkesch_2014}}]\label{thm:meromorphic-cont}
  Suppose the polytope $P=P(f_{1})+\ldots+P(f_{k})$ is full-dimensional. Then for every
  $(s_{0},c_{0})\in M_{\mathbb{C}}\times \mathbb{C}^{k}$ the Mellin transform
  $\mathcal{M}(f_{i},g,s,c)$ can be expressed as a sum of the form
  \begin{equation*}
    \mathcal{M}(f_{i},g,s,c) = \sum_{\beta}L_{\beta}(s,c)\mathcal{M}(f_{i},g_{\beta},s,c+n_{\beta}),
  \end{equation*}
  for certain Laurent polynomials $g_{\beta}$ and $n_{\beta}\in \mathbb{Z}_{\ge 0}^{k}$, such that
  the Mellin transforms on the right hand side are convergent in a neighbourhood of $(s_{0},c_{0})$.
  The functions $L_{\beta}(s,c)$ are rational functions of $(s,c)$ with simple poles along divisors
  of the form
  \begin{equation*}
    \{\langle s, u_{\rho} \rangle - d_{\rho}(c) + d_{\rho}(g) = -m\}
  \end{equation*}
  for $m\in \mathbb{N}$. Thus the Mellin transform can be expressed as
  \begin{equation*}
    \mathcal{M}(f_{i},g,s,c)=\Phi(s,c)
    \prod_{\rho\in\Sigma_{P}(1)}\Gamma(\langle s, u_{\rho} \rangle -d_{\rho}(c)+d_{\rho}(g)),
  \end{equation*}
  where $\Phi(s,c)$ is entire on $M_{\mathbb{C}}\times \mathbb{C}^{k}$.
\end{thm}

\begin{proof}[Proof Sketch.]
  Fix $(s_{0},c_{0})\in M_{\mathbb{C}}\times \mathbb{C}^{k}$ and let
  \begin{equation*}
    a_{\rho} = -\min(\lceil{  \Real(\langle s_{0}, u_{\rho} \rangle - d_{\rho}(c_{0})+d_{\rho}(g))}\rceil-1,0)
  \end{equation*}
  For each $\rho\in\Sigma_{P}(1)$ we partially integrate (at most) $a_{\rho}$-times in the direction
  $u_{\rho}$. This expresses the Mellin transform as a sum
  \begin{equation*}
    \mathcal{M}(f_{i},g,s,c) = \sum_{\beta}L_{\beta}\mathcal{M}(f_{i},g_{\beta},s,c+n_{\beta}),
  \end{equation*}
  where $L_{\beta}$ is the product of the rational factors $(\langle s, u_{\rho} \rangle -
  d_{\rho}(c) + d_{\rho}(g) +m)^{-1}$ introduced by the partial integrations. One can check that
  these poles are all distinct, so that $L_{\beta}$ is a rational function with simple poles as
  prescribed above.

  The Mellin transforms on the right hand side are guaranteed to converge in a neighbourhood of
  $(s_{0},c_{0})$, since each partial integration in the $\rho$ direction improves the convergence
  in this direction by at least one unit. The Gamma functions $\Gamma(\langle s, u_{\rho} \rangle
  -d_{\rho}(c)+d_{\rho}(g))$ have simple poles along every hypersurface of the form
  \begin{equation*}
    \{\langle s, u_{\rho} \rangle - d_{\rho}(c) + d_{\rho}(g)\in -\mathbb{Z}_{+}\}
  \end{equation*}
  Hence dividing the above expression for $\mathcal{M}(f_{i},g,s,c)$ by
  \begin{equation*}
    \prod_{\rho\in\Sigma_{P}(1)}\Gamma(\langle s, u_{\rho} \rangle -d_{\rho}(c)+d_{\rho}(g))
  \end{equation*}
  cancels the simple poles and gives an entire function $\Phi(s,c)$.
\end{proof}

\paragraph{Sector decomposition.}
\label{sec:sect-decomp}

In applications to quantum field theory, one often wants to compute the analytic continuation of
$\mathcal{M}(f_{i},s,c)$ along a single parameter (the dimension of spacetime), i.e. one wants to
compute the restriction of $\mathcal{M}(f_{i},s,c)$ to a line $l\subset M_{\mathbb{C}}\times
\mathbb{C}^{k}$. The corresponding function $M(d)=\mathcal{M}(f_{i},s(d),c(d))$ is a meromorphic
function of $d\in \mathbb{C}$ and the main goal is to compute the coefficients of a Laurent
expansion around a fixed pole $d_{0}\in \mathbb{C}$.

For this purpose, Binoth and Heinrich \cite{Binoth_2000} introduced a recursive strategy, which
iteratively decomposes (a suitable compactification) of the integration domain
$T_{N}(\mathbb{R^{+}})$ into cubical sectors and performs blow-ups along coordinate subspaces until
the integral in every sector is of the form
\begin{equation*}
  I_{\alpha} = \int_{[0,1]^{k}}x^{s_{\sigma}(d)}\prod_{i=1}^{k}\tilde f^{c_{i}(d_{i})}_{i},
\end{equation*}
such that $\tilde f_{i}$ does not vanish along the coordinate subspaces $x_{i}=0$. The analytic
continuation in $d$ can then be computed by a simple Taylor expansion.

The original strategies of Binoth and Heinrich had the drawback, that the recursion did not always
terminate. This fault was corrected by Bogner and Weinzierl in \cite{Bogner_2008}, where ideas from
the resolution of singularities were used to devise strategies guaranteed to succeed. Unfortunately,
these strategies often result in a large number of sectors, which greatly impacts the time needed
for numerical computations.

The results of the last section suggest that one should try to find a strategy which is adapted to
the Newton polytope of $f_{1}\cdots f_{k}$. Such a strategy has indeed been described by Kaneko and
Ueda in \cite{Kaneko_2010} and it fits in nicely with our toric point of view. One of the main
results of \cite{Kaneko_2010} can be rephrased as follows.

\begin{prop}\label{prop:sec-decom}
  Let $X_{\Sigma}$ be a complete simplicial toric variety of dimension $n$. For $\sigma\in\Sigma(n)$
  let
  \begin{equation*}
    I_{\sigma} := [0,1]^{n}//G_{\sigma}\subset \mathbb{R}_{+}^{n}//G_{\sigma}= U_{\sigma}(\mathbb{R}^{+}).
  \end{equation*}
  Then
  \begin{equation*}
    X_{\Sigma}(\mathbb{R}_{+}) = \bigcup_{\sigma\in\Sigma(n)}I_{\sigma}
  \end{equation*}
  and the intersections $I_{\sigma}\cap I_{\sigma'}$ for $\sigma\neq\sigma'$ have measure zero.
\end{prop}
\begin{proof}
  The complement of $T_{N}(\mathbb{R^{+}})\subset X_{\Sigma}(\mathbb{R}_{+})$ has real codimension
  one, so it is enough to show that the restrictions
  \begin{equation*}
    I^{\circ}_{\sigma} = I_{\sigma}\cap T_{N}(\mathbb{R^{+}}) \cong (0,1]^{n}//G_{\sigma}
  \end{equation*}
  cover $T_{N}(\mathbb{R^{+}})$ and intersect in a set of measure zero. The map
  \begin{equation*}
    L : T_{N}(\mathbb{R^{+}})\rightarrow N_{\mathbb{R}},\quad L(t) = -\Log(t)
  \end{equation*}
  is a diffeomorphism. We claim that $L$ identifies $I^{\circ}_{\sigma}$ with the cone
  $\sigma\subset \mathbb{R}^{n}$. Let $\sigma$ have lattice generators $u_{1},\ldots,u_{n}$. In the
  coordinates of $\sigma$ we have $t_{j}=\prod_{i=1}^{n}x_{i}^{\langle e_{j}, u_{i} \rangle}$ and
  \begin{align*}
    L(t(x)) &= \sum_{j=1}^{n}-\log(t_{j}(x))e_{j}\\
            &= \sum_{i,j=1}^{n}-\log(x_{i})\langle e_{j}, u_{i} \rangle e_{j}\\
            &= \sum_{i=1}^{n}-\log(x_{i})u_{i}.
  \end{align*}
  Since $I^{\circ}_{\sigma}$ is defined by the inequalities $0<x_{i}\le 1$, we get
  \begin{equation*}
    L(I^{\circ}_{\sigma})=¸\pos(u_{1},\ldots,u_{n})=\sigma.
  \end{equation*}
  But $X_{\Sigma}$ is a complete fan, so the cones $\sigma\in\Sigma(n)$ cover $\mathbb{R}^{n}$ and
  intersect in a set of codimension one. The same must then be true for the $I_{\sigma}^{\circ}$.
\end{proof}

Combining the above proposition with proposition \ref{prop:simplicial-refinement} gives the main
result of \cite{Kaneko_2010}:

\begin{col}[\cite{Kaneko_2010}]\label{col:sec-decomp}
  Let $X_{\Sigma}$ be a simplicial refinement of $P=P(f_{1}\cdots f_{k})$. Then the Mellin transform
  can be decomposed as
  \begin{equation*}
    M(f,s,c) = \sum_{\sigma\in\Sigma(n)}\int_{[0,1]^{n}} x_{\sigma}^{s_{\sigma}}\prod_{i=1}^{k}(f_{\sigma,i}(x_{\sigma}))^{-c_{i}} \df x_{\sigma},
  \end{equation*}
  where:
  \begin{enumerate}
  \item $x_{\sigma}=(x_{\sigma,1},\ldots,x_{\sigma,n})$ are the coordinates associated to the
    maximal cone $\sigma=\pos(u_{1},\ldots,u_{n})\in \Sigma$.
  \item The monomial $x_{\sigma}^{s_{\sigma}}$ is given by
    \begin{align*}
      x_{\sigma}^{s_{\sigma}} = \prod_{j=1}^{n}x_{\sigma,j}^{\langle s - \sum_{i}c_{i}m_{\sigma,i}, u_{j} \rangle -1}
      = \prod_{j=1}^{n}x_{\sigma,j}^{\langle s, u_{j}\rangle -d_{u_{j}}(c)-1}
    \end{align*}
    with $m_{\sigma,i}$ the common minimum of the functions $\langle m, u_{j} \rangle$ on
    $P(f_{i})$.
  \item The polynomials $f_{\sigma,i}(x_{\sigma})=f_{i}(x_{\sigma})x_{\sigma }^{-m_{\sigma,i}}$ are
    regular and non-vanishing on $[0,1]^{n}$.
  \end{enumerate}
\end{col}

\begin{proof}
  The previous proposition shows that we can write the integral as a sum over the integration
  domains $I_{\sigma}=[0,1]^{n}//G_{\sigma}$. Over this domain the integral becomes
  \begin{align*}
    \int_{I_{\sigma}}t^{s}\prod_{i=1}^{k}(f_{i}(t))^{-c_{i}} \frac{\df t}{t}
    &= \int_{I_{\sigma}}\prod_{\rho\in\Sigma(1)}x_{\rho}^{\langle s, u_{\rho} \rangle -1}
      \prod_{i=1}^{k}F_{i}(x)^{-c_{i}} \Omega_{\Sigma}{\bigg\rvert}_{U_{\sigma}} \\
    &= \frac{u_{\sigma}}{|G_{\sigma}|}\int_{[0,1]^{n}} \prod_{j=1}^{n}x_{j}^{\langle s, u_{j} \rangle -1}
      \prod_{i=1}^{k}f_{i}(x_{\sigma})^{-c_{i}} \df x_{\sigma} \\
    &= \frac{u_{\sigma}}{|G_{\sigma}|}\int_{[0,1]^{n}} \prod_{j=1}^{n}x_{j}^{\langle s-\sum_{i}c_{i}m_{\sigma,i}, u_{j} \rangle -1}
      \prod_{i=1}^{k}f_{\sigma,i}(x_{\sigma})^{-c_{i}} \df x_{\sigma}.
  \end{align*}
  and it follows from the proof of proposition \ref{prop:simplicial-refinement} that
  $f_{\sigma,i}(x_{\sigma})$ is regular and non-vanishing on $[0,1]^{n}$. From the discussion in
  section \ref{sec:TV:divisor-coordinate-ring} we can assume that $u_{\sigma}>0$ and we have
  \begin{equation*}
    u_{\sigma}=\det(u_{1},\ldots,u_{n}) = |\Cl(U_{\sigma})| = |G_{\sigma}|,
  \end{equation*}
  so the factor $\frac{u_{\sigma}}{|G_{\sigma}|}$ cancels.
\end{proof}

\begin{remark}
  The integral over $I_{\sigma}$ converges for $(s,c)\in M_{\mathbb{C}}\times \mathbb{C}^{k}$
  satisfying
  \begin{equation*}
    \Real\langle s, u_{\rho} \rangle > \Real d_{\rho}(c)
  \end{equation*}
  for all $\rho\in \sigma(1)$. If $r:=\Real(c)>0$, this means that $\Real(s)$ lies in the
  interior of the convex cone $m_{\sigma}(r)+\sigma^{\vee}$, where
  $m_{\sigma}(r)=\sum_{i}r_{i}m_{\sigma,i}$ is the vertex of $P(r)$ corresponding to $\sigma$. Hence
  the convergence region for a single sector is always nonempty, even if $P$ is not
  full-dimensional. In fact, if $r>0$ is fixed, then
  \begin{equation*}
    \bigcap_{\sigma\in\Sigma(1)} m_{\sigma}(r) + \sigma^{\vee} =  P(r).
  \end{equation*}
  Thus the integrals $I_{\sigma}$ convergence simultaneously for $s\in M_{\mathbb{C}}$ if and only
  if $s\in \Int(\bigcap_{\sigma\in\Sigma(1)} m_{\sigma}(r) + \sigma^{\vee}) = \Int P(r)$, which
  recovers our earlier result.
\end{remark}

\section{Feynman integrals}
\label{sec:feynman-graphs}

Let us finally apply our work to the investigation of Feynman integrals. We will work will scalar
Feynman graphs with generic euclidean kinematics. The parametric representation then expresses the
amplitude as a Mellin transform to which our previous results apply. We will show that this gives
two equivalent ways to rigourously construct the dimensionally regularized amplitudes.

\paragraph{Feynman graphs. }
We will consider a \emph{graph} $G$ to consist of a triple
\begin{equation*}
  G = (E_{G},V_{G},\partial)
\end{equation*}
of finite sets of edges $E_{G}$ and vertices $V_{G}$, together with a map
\begin{equation*}
  \partial:E_{G}\rightarrow \Sym^{2}V_{G}=V_{G}\times V_{G}/\mathbb{Z}_{2},
\end{equation*}
mapping an edge to its endpoints. This definition allows multiple edges and loops, but our graphs
will not have external half-edges. An edge $e\in E_{G}$ is called a \emph{selfloop} if
$\partial(e)=(v,v)$. A subgraph $\gamma\subset G$ is given by subsets $E_{\gamma}\subset E_{G},
V_{\gamma}\subset V_{G}$, such that $\partial(E_{\gamma})\subset \Sym^{2}V_{\gamma}$.

Every graph has an obvious geometric realization as a one-dimensional CW-complex, so that we can
speak about topological notions like connectedness and simply-connectedness. In particular, we
denote by $h^{0}(G)$ and $h^{1}(G)$ the first and second Betti numbers of (the geometric realization
of) $G$. For a connected graph $G$, we call a connected subgraph $T\subset G$ a \emph{spanning tree}
if $V_{T}=V_{G}$ and $h^{1}(T)=0$. Note that these are precisely the maximal simply-connected
subgraphs of $G$.

A spanning 2-tree is a simply-connected subgraph $F\subset G$, with $V_{F}=V_{G}$ and exactly two
connected components $F=T_{1}\cup T_{2}$. Every spanning 2-tree is obtained from a spanning tree by
deleting an edge.

Every subset $I\subset E_{G}$ gives the \emph{edge subgraph} $\gamma\subset G$, where $E_{\gamma}=I$
and $V_{\gamma}$ consists of all vertices incident to an edge in $I$. We almost exclusively deal
with edge subgraphs, so we will often identify an edge subgraph with its set of edges. Notable
exceptions are spanning 2-trees, where it is important to allow isolated vertices.

A \emph{Feynman graph} is a graph together with distinguished (possibly empty) sets of external
vertices $V^{ext}_{G}\subseteq V$ and massive edges $E_{G}^{M}\subseteq E_{G}$. To every external
vertex $v\in V^{G}$ we associate an inflowing external momentum $q_{v}\in \mathbb{C}^{D}$ and to
every massive edge $e\in E^{M}_{G}$, a mass $m_{e}\in\mathbb{C}\backslash \{0\}$. The external
momenta are additionally subject to momentum conversation in each connected component: If $G=\cup
G_{i}$ is the decomposition of $G$ into connected components, then
\begin{equation*}
  \sum_{v\in V^{ext}_{G_{i}}} q_{v} = 0,
\end{equation*}
for all $i$.

Now suppose $G$ is a connected Feynman graph. Choosing an orientation for each edge gives $G$ the
structure of a one-dimensional cell complex. Then a truncated part of the cellular chain complex
gives the exact sequence

\begin{center}\label{eq:graph-exact-sequence}
  \begin{tikzcd}
    0 \arrow{r} & H_{1}(G,\mathbb{Z}) \arrow{r}{i} & \mathbb{Z}^{E_{G}} \arrow{r}{\partial} & V^0_G
    \arrow{r} & 0
  \end{tikzcd}
\end{center}

Here, $\partial$ is the boundary map and
\begin{equation*}\label{eq:mom-conv}
  V^{0}_{G} := \{(n_{v})\in \mathbb{Z}^{V_{G}}\ \rvert\ \sum n_{v}=0\}
\end{equation*}
is the image of $\partial$, imposing overall momentum conservation. The external momentum $q$ is
then naturally an element of $V^{0}_{G}\otimes \mathbb{C}^{D}$. We can choose a section
$B:V^{0}_{G}\rightarrow \mathbb{Z}^{E_{G}}$ of $\partial$, since the homology $H_{1}(G,\mathbb{Z})$
is free. In the physics literature, it is customary to choose a spanning tree $T\subset G$ and then
defining a section by
\begin{equation*}
  B_{T}: V^{0}_{G}=V^{0}_{T}\cong \mathbb{Z}^{E_{T}}\hookrightarrow \mathbb{Z}^{E_{G}}.
\end{equation*}

Let $p\mapsto p_{e}$ be the projection $\mathbb{C}^{E}\otimes \mathbb{C}^{D}\rightarrow
\mathbb{C}^{D}$ to the momentum flowing through the edge $e$. To each edge we associate the affine
quadric
\begin{equation*}
  P_{e}:\mathbb{C}^{D}\rightarrow \mathbb{C},\quad P_{e}(k)=(p_{e}^{2}+m^{2}_{e}),
\end{equation*}
where $p_{e}^{2}=\sum_{i=1}^{D}(p_{e}^{i})^{2}$ denotes the squared euclidean metric. Its inverse
$P_{e}^{-1}$ gives the (scalar) propagator. As outlined in the introduction, we will work with
analytically regularized integrals, which means raising each propagator to a complex power
$\lambda_{e}\in \mathbb{C}$.

\begin{defin}
  The (formal) euclidean Feynman amplitude $I_{G}(\lambda,p,m)$ is given by
  \begin{equation*}
    I_{G}(\lambda,D,q,m) := \int_{H_{1}(G,\mathbb{R}^{D})}\prod_{e\in E_{G}}(P_{e}(k+B(q)))^{-\lambda_{e}}\df\mu,
  \end{equation*}
  where $\df \mu = \frac{\df^{Dl}k}{\pi^{Dl/2}}$ is a convenient multiple of the Lebesgue measure.
\end{defin}

\begin{remark}
  In general quantum field theories, (e.g. gauge theories) the Feynman rules also give polynomials
  of invariant scalar products in the numerator. But it is well known that these ``tensor''
  integrals can be expressed as a linear combination of scalar Feynman integrals with shifted values
  of the dimension and propagator decorations, see e.g. \cite{Tarasov_1996}. Since this is naturally
  part of our approach anyway, we do not lose any generality if this reduction is understood.
\end{remark}

\paragraph{Parametric representation.}
\label{sec:schwinger}

In order to apply the results of section \ref{sec:analyt-dimens-regul}, we want to express $I_{G}$
as a suitable Mellin transform. This will only work if the external momenta and masses are
sufficiently generic.

\begin{defin}\label{defin:generic-kinematics}
  A Feynman Graph $G$ has generic euclidean kinematics, if
  \begin{align*}
    \Real\left( \sum_{i\in I}q_{i} \right)^{2} &> 0 \\
    \Real\left( \sum_{i\in I}q_{i} \right)^{2} + \Real(m^{2}_{e}) &> 0
  \end{align*}
  for all proper subsets $I\subsetneq V^{ext}_{G}$ and massive edges $e\in E^{m}_{G}$.
\end{defin}
We will assume that $G$ has generic euclidean kinematics from now on.
When none of the $P_{e}$ vanish, we can use the identity
\begin{equation*}
  \frac{1}{p^{\lambda}}  = \int_{0}^{\infty}\frac{\alpha^{\lambda-1}}{\Gamma(\lambda)} e^{-\alpha p}\df \alpha, \quad \Real(p) > 0,
\end{equation*}
to write the integrand of the momentum-space amplitude as
\begin{equation*}
  \prod_{e\in E_{G}}P_{e}^{-\lambda_{e}}(p)=\int_{[0,\infty]^{E}}\prod_{e\in E_{G}}\frac{\alpha^{\lambda_{e}-1}}{\Gamma(\lambda_{e})}d\alpha_{e}e^{-\sum\alpha_{e}P_{e}(p)}.
\end{equation*}

Formally exchanging the integration variables gives the amplitude as
\begin{align*}
  I_{G}(\lambda,D,q,m)
  &= \int_{[0,\infty]^{E_{G}}}\prod_{e\in E_{G}}\frac{\alpha^{\lambda_{e}-1}}{\Gamma(\lambda_{e})}d\alpha_{e}\int_{H_{1}(G,\mathbb{R}^{D})}e^{-\sum\alpha_{e}P_{e}(k+B(q))}\df \mu \\
  &=: \int_{[0,\infty]^{E_{G}}}\prod_{e\in E_{G}}\frac{\alpha^{\lambda_{e}-1}}{\Gamma(\lambda_{e})}\df \alpha_{e}A_{G}(\lambda,D,q,m,\alpha).
\end{align*}

The integrand $A_{G}$ can be computed by reducing it to a gaussian integral. The answer will involve
the following graph polynomials:

\begin{defin}
  The \emph{first Symanzik} polynomial of the connected graph $G$ is
  \begin{equation*}
    \psi_{G}:= \sum_{T}\prod_{e\notin T}\alpha_{e},
  \end{equation*}
  where the sum is over all spanning trees of $G$. The (massless) \emph{second Symanzik} polynomial of $G$
  with external momentum $q\in \mathbb{C}^{E}\otimes \mathbb{C}^{D}$ is
  \begin{equation*}
    \varphi_{G}(q,\alpha) = \sum_{T_{1}\cup T_{2}}q^{2}_{T_{1}}\prod_{e\notin T_{1}\cup T_{2}}\alpha_{e},
  \end{equation*}
  where the first sum is over spanning two-forests and
  \begin{equation*}
    q_{T_{1}} = \sum_{v\in V(T_{1})}q_{v}
  \end{equation*}
  is the total momentum flowing through $T_{1}$. The full second Symanzik polynomial of $G$ is
  defined as
  \begin{equation*}
    \Phi_{G}(\alpha,q,m) = \varphi_{G}(\alpha,q)+\left( \sum_{i\in E_{G}}\alpha_{i}m^{2}_{i} \right)\psi_{G}(\alpha)
  \end{equation*}
\end{defin}

\begin{remark}
  By momentum conservation we have
  \begin{equation*}
    q_{T_{1}}^{2} = (-q_{T_{2}})^{2} = q_{T_{2}}^{2}.
  \end{equation*}
  Hence the above definition is unambiguous.
\end{remark}

The following theorem is then well-known (See e.g. \cite{Panzer:2015ida},
\cite{itzykson1985quantum}).

\begin{thm}
  The integrand $A_{G}$ can be calculated as
  \begin{equation*}
    A_{G}(\lambda,D,\alpha,q,m) = \frac{e^{-\frac{\Phi_{G}}{\psi_{G}}}}{\psi_{G}^{D/2}}.
  \end{equation*}
\end{thm}

Now consider the diffeomorphism
\begin{align*}
  F:(0,\infty)\times \Delta^{E_{G}} &\cong \mathbb{R}^{E_{G}}_{+}\backslash\{0\} \\
  (t,\alpha) \mapsto t\alpha
\end{align*}
A simple calculation shows that the standard volume form $\prod_{e}\df\alpha_{e}$ pulls back to
\begin{equation*}
  F^{*}\left( \prod_{e}\df^{n}\alpha_{e} \right) = t^{E-1}\df t \wedge \Omega_{\Delta_{E_{G}}},
\end{equation*}
where
\begin{equation*}
  \Omega_{\Delta_{E_{G}}} = \sum_{i=1}^{E}(-1)^{i-1}a_{i}\df a_{1}\wedge \ldots \widehat{\df a_{i}} \ldots \wedge \df a_{E}\bigg\rvert_{\Delta_{E_{G}}}.
\end{equation*}
and we have fixed some ordering $E_{G}\cong\{1,\ldots,E\}$. Now let
\begin{equation*}
  \omega_{G} = \sum_{e\in E_{G}}\lambda_{e} - \frac{D}{2}h^{1}_{G}=:-sd_{G},
\end{equation*}
where $sd_{G}$ is the superficial degree of divergence of $G$.
Performing the integral over $t$ gives
(still formally)
\begin{align*}
  I_{G}(\lambda,D,q,m) &= \int_{[0,\infty]^{E}}\prod_{e\in E_{G}}\frac{\alpha^{\lambda_{e}-1}}{\Gamma(\lambda_{e})}\df \alpha_{e}A_{G}(\lambda,D,q,m,\alpha)\\
                       &= \int_{\Delta_{E_{G}}}\prod_{e\in E_{G}}\frac{\alpha^{\lambda_{e}-1}}{\Gamma(\lambda_{e})}\frac{\Omega_{\Delta_{E_{G}}}}{\psi_{G}^{D/2}}\int_{0}^{\infty}t^{\omega_{G}-1}e^{-t \frac{\Phi}{\psi}}\df t\\
                       &= \int_{\Delta_{E_{G}}}\prod_{e\in E_{G}}\frac{\alpha^{\lambda_{e}-1}}{\Gamma(\lambda_{e})}\frac{\Omega_{\Delta_{E_{G}}}}{\psi_{G}^{D/2}}\Gamma(\omega_{G})\left( \frac{\psi_{G}}{\Phi_{G}} \right)^{\omega_{G}}.
\end{align*}
There is a natural homeomorphism
\begin{equation*}
  P^{E_{G}}(\mathbb{R}^{+}) \rightarrow \Delta^{E_{G}},\quad [\alpha_{1}:\ldots:\alpha_{E_{G}}]\mapsto \left( \frac{\alpha_{i}}{\sum_{i\in E_{G}}\alpha_{i}} \right),
\end{equation*}
which identifies $\Omega_{\Delta_{E_{G}}}$ with the volume form $\Omega_{P^{E_{G}}}$ of $P^{E_{G}}$
constructed in section \ref{sec:TV:divisor-coordinate-ring}. We can then express the above result as
follows.
 
\begin{col}
  The amplitude can be expressed as the (possibly still divergent) projective integral
  \begin{equation*}
    I_{G}(\lambda,D,q,m) = \Gamma(\omega_{G})\int_{P^{E_{G}}(\mathbb{R}^{+})}\prod_{e\in E_{G}}\frac{\alpha^{\lambda_{e}-1}}{\Gamma(\lambda_{e})}\left( \frac{\psi_{G}}{\Phi_{G}} \right)^{\omega_{G}}\frac{\Omega_{P^{E_{G}}}}{\psi_{G}^{D/2}}.
  \end{equation*}
\end{col}

We can interpret the above representation of $I_{G}$ as an integral over the open simplex
\begin{equation*}
  T_{N_{E_{G}}}(\mathbb{R}^{+})\subset P^{E_{G}}(\mathbb{R}^{+}),
\end{equation*}
which expresses it as Mellin transform.
For generic euclidean kinematics, the coefficients of $\Phi_{G}$ are contained in an open
strongly convex cone. Hence we are finally in a position to apply the result of section
\ref{sec:analyt-dimens-regul}.

\begin{remark}
  We will soon see that (under mild conditions on the graph G) we can choose the $\lambda_e$ such
  $\Real(\omega_{G})>0$ and the above expression for $I_{G}$ is absolutely convergent for generic
  euclidean momenta. In this case, we get an equality of absolutely convergent integrals
  \begin{align*}
    I_{G}(\lambda,D,q,m)
    &= \int_{H_{1}(G,\mathbb{R}^{D})}\prod_{e\in E_{G}}(P_{e}(k + B(p)))^{-\lambda_{e}}\df\mu\\
    &= \Gamma(\omega_{G})\int_{P^{E_{G}}(\mathbb{R}^{+})}\prod_{e\in E_{G}}\frac{\alpha^{\lambda_{e}-1}}{\Gamma(\lambda_{e})}\left( \frac{\psi_{G}}{\Phi_{G}} \right)^{\omega_{G}}\frac{\Omega_{P^{E_{G}}}}{\psi_{G}^{D/2}}.
  \end{align*}
\end{remark}

\paragraph{The Newton polytope of a Feynman graph.}

The results of section \ref{sec:analyt-dimens-regul} show that, to understand the convergence domain
of the integral $I_{G}$, we should understand the polytope $P(\psi_{G}\Phi_{G})$.

\begin{defin}
  The Newton polytope $P(\psi_{G}\Phi_{G})$ of a Feynman graph $G$ is called the Feynman polytope of
  $G$ and denoted by $P_{G}$.
\end{defin}

Let us first recall the factorization formulas of Brown (\cite{Brown_2017}). If $G$ has some massless
edges (and thus possible IR divergences) a special role is played by subgraphs, which carry all the
kinematics, in the following sense:

\begin{defin}
  An edge subgraph $\gamma\subset G$ containing all external vertices of $G$ in a single connected
  component will be called \emph{momentum spanning}. If $\gamma$ additionally contains all massive
  edges, then it will be called \emph{mass-momentum spanning} (m.m. for short).
\end{defin}
For an m.m. subgraph $\gamma$, we set $V^{ext}_{\gamma}=V^{ext}_{G}$ and $E^{M}_{\gamma}=E^{M}_{G}$
and the kinematics of $\gamma$ are the same as those of $G$. Otherwise we consider $\gamma$ to be
scaleless, i.e. $V^{ext}_{\gamma}=E^{M}_{\gamma}=\emptyset$. If $\gamma$ is a possibly disconnected
subgraph, we define the quotient $G/\gamma$ by contracting every connected component to a vertex.
The kinematics of $G/\gamma$ are inherited from $G$ in the obvious way. This implies that $G/\gamma$
has nontrivial kinematics if and only if $\gamma$ is not mass-momentum spanning.

In this way, we can associate to every subgraph $\gamma\subset G$ the ``flat deformation''
\begin{equation*}
  G|\gamma = \gamma \cup G/\gamma,
\end{equation*}
where exactly one of $\gamma$ and $G/\gamma$ has nontrivial kinematics. For a possibly disconnected
Feynman graph $\Gamma$ as above, we generalize the definitions of the Symanzik polynomials as
follows: If $\Gamma=\cup_{i=1}^{k} \Gamma_{i}$ is the disjoint union of connected graphs
$\Gamma_{i}$, then we set
\begin{align*}
  \psi_{\Gamma} &= \sum_{i=1}^{k}\psi_{\Gamma_{i}}\\
  \varphi_{\Gamma} &= \sum_{i=1}^{k}\varphi_{\Gamma_{i}}\prod_{j\neq i}\psi_{\Gamma_{j}}\\
  \Phi_{\Gamma} &= \sum_{i=1}^{k}\Phi_{\Gamma_{i}}\prod_{j\neq i}\psi_{\Gamma_{j}}.
\end{align*}

It will also be convenient to define $\delta^{m}_{\gamma}$ to be 1 if $\gamma$ is momentum-spanning
and $0$ otherwise. Similarly, we set $\delta^{mm}_{\gamma}$ to be 1 if $\gamma$ is mass-momentum
spanning and $0$ otherwise. With this notation, we can formulate the factorization formulas as
follows:

\begin{prop}[\cite{Brown_2017}]\label{prop:factor}
  Let $G$ be a connected Feynman graph and $\gamma\subset G$ a subgraph with connected components
  $\gamma_{0},\ldots,\gamma_{k}$.
  \begin{enumerate}
  \item There are polynomials $R^{\psi}_{G|\gamma},R^{\varphi}_{G|\gamma}$ and
    $R^{\Phi}_{G|\gamma}$, such that
    \begin{align*}
      \psi_{G} &= \psi_{G|\gamma} + R^{\psi}_{G|\gamma}\\
      \varphi_{G} &= \varphi_{G|\gamma} + R^{\varphi}_{G|\gamma} \\
      \Phi_{G} &= \Phi_{G|\gamma} + R^{\Phi}_{G|\gamma}.
    \end{align*}
    The degree $\deg_{\gamma}(R^{\cdot}_{G|\gamma})$ of the rest terms in the variables
    $(\alpha_{e})_{e\in\gamma}$ satisfies
    \begin{align*}
      \deg_{\gamma}(R^{\psi}_{G|\gamma}) &> \deg_{\gamma}(\psi_{G|\gamma})
                                           = h^{1}_{\gamma} \\
      \deg_{\gamma}(R^{\varphi}_{G|\gamma}) &> \deg_{\gamma}(\varphi_{G|\gamma})
                                           = h^{1}_{\gamma} + \delta^{m}_{\gamma}\\
      \deg_{\gamma}(R^{\Phi}_{G|\gamma}) &> \deg_{\gamma}(\Phi_{G|\gamma})
                                           = h^{1}_{\gamma}+\delta^{mm}_{\gamma}
    \end{align*}

  \item The polynomial $R^{\psi}_{G|\gamma}$ vanishes if and only if, for any spanning tree
    $T\subset G$, the graphs $\gamma_{i}\cap T$ are connected.
  \item Suppose $\gamma$ is not momentum spanning. Then the polynomial $R^{\varphi}_{G|\gamma}$
    vanishes if and only if the intersections $\gamma_{i}\cap F$ are connected for any spanning
    $2$-tree $F=T_{1}\cup T_{2}$ with $(q_{T_{1}})^{2}\neq 0$.
  \item Suppose $\gamma$ is momentum spanning and all external vertices are contained in the
    component $\gamma_{0}$. Then $R^{\varphi}_{G|\gamma}$ vanishes if and only if, for any spanning
    $2$-tree $F=T_{1}\cup T_{2}$ with $(q_{T_{1}})^{2}\neq 0$, the graphs $F\cap \gamma_{i}$ are connected for
    $i>0$ and $F\cap \gamma_{0}$ has exactly two connected components.
  \item The polynomial $R^{\Phi}_{G|\gamma}$ is given by
    \begin{equation*}
      R^{\Phi}_{G|\gamma} = R_{G|\gamma}^{\varphi} + \psi_{G|\gamma}\left( \sum_{e\in E^{M}_{G}\cap E_{\gamma}}m^{2}_{e}\alpha_{e} \right) + R^{\psi}_{G|\gamma}\left( \sum_{e\in E^{M}_{G}}m^{2}_{e}\alpha_{e} \right)
    \end{equation*}
    if $\gamma$ is not mass-momentum spanning and by
    \begin{equation*}
      R^{\Phi}_{G|\gamma} = R_{G|\gamma}^{\varphi} +  R^{\psi}_{G|\gamma}\left( \sum_{e\in E^{M}_{G}}m^{2}_{e}\alpha_{e} \right)
    \end{equation*}
    if $\gamma$ is mass-momentum spanning.
  \end{enumerate}
\end{prop}

\begin{proof}[Proof sketch]
  We sketch the proof and refer to \cite{Brown_2017} for more details. For $k=1,2$, consider the
  subset $\mathcal{T}^{k}_{\gamma}$ of spanning $k$-trees of $G$, such that the intersections
  $T\cup\gamma_{i}$ are connected for all $T\in \mathcal{T}^{k}_{\gamma}$. This set is always
  nonempty (\cite[Lemma 2.1.]{Brown_2017}) and its elements are exactly those spanning $k$-trees,
  such that $\sum_{i}|T\cap \gamma_{i}|$ is maximal. The corresponding monomial $\alpha^{S}$ for
  $S=E_{G}\backslash T$ is then minimal in the variables $(\alpha_{e})_{e\in \gamma}$. Decomposing
  the sum over $k$-trees into a sum over $\mathcal{T}^{k}_{\gamma}$ and its complement gives the
  decompositions
  \begin{align*}
    \psi_{G}&=\psi_{\gamma}\psi_{G/\gamma} + R^{\psi}_{G|\gamma} \\
    \varphi_{G}&=\psi_{\gamma}\varphi_{G/\gamma} + R^{\varphi}_{G|\gamma},
  \end{align*}
  and $R^{\psi}_{G|\gamma}$ (resp. $R^{\varphi}_{G|\gamma}$) vanishes, if and only if
  $\mathcal{T}^{1}_{\gamma}$ (resp. $\mathcal{T}^{2}_{\gamma}$) consists of all spanning trees
  (resp. spanning $2$-trees). If $\gamma$ is not momentum spanning, then
  $\varphi_{G|\gamma}=\psi_{\gamma}\varphi_{G/\gamma}$ and the proposition follows. Now suppose
  $\gamma$ is momentum spanning, such that all external vertices are contained in the component
  $\gamma_{0}$. Then the polynomial $\varphi_{G/\gamma}$ vanishes and we have to use a different
  decomposition. In this case, we define $\mathcal{\tilde T}^{2}_{\gamma}$ to be those spanning
  $2$-trees $F$, such that $F\cap \gamma_{i}$ are trees for $i>0$ and $F$ splits $\gamma_{0}$ into
  two connected components. Splitting the sum over all $2$-trees into a sum over $\mathcal{\tilde
    T}^{2}_{\gamma}$ and its complement as above, gives the decomposition
  \begin{equation*}
    \varphi_{G} = \varphi_{\gamma_{0}}\psi_{\gamma_{1}}\cdots\psi_{\gamma_{k}}\psi_{G/\psi} + R^{\varphi}_{G|\gamma}.
  \end{equation*}
  The first term is just $\varphi_{G|\gamma}$ and the second term vanishes if and only if the
  complement of $\mathcal{\tilde T}^{2}_{\gamma}$ is empty. This establishes the first $4$ points.
  The expression for $R^{\Phi}_{G|\gamma}$ follows immediately from the decompositions of $\psi_{G}$
  and $\varphi_{G}$.
\end{proof}

\begin{col}
  The face of $P_{G}$ corresponding to the weight vector $e^{\gamma}$ is
  \begin{equation*}
    F_{e^{\gamma}}P_{G} = P_{\gamma} \times P_{G/\gamma}.
  \end{equation*}
\end{col}

In the above Corollary, we have identified an edge subgraph $\gamma\subseteq G$ with its set of
edges $E_{\gamma}\subseteq E_{G}$. Using this convention, we define the subset function
$s_{G}:2^{E_{G}}\rightarrow \mathbb{Z}$ by
\begin{equation*}
  s_{G}(\gamma) = 2h^{1}(\gamma) + \delta^{mm}_{\gamma}.
\end{equation*}

The above proof shows in particular that
\begin{equation*}
  h^{1}(\gamma) = |E_{\gamma}| - |E_{\gamma}\cap E_{T}|,
\end{equation*}
where $T$ is a spanning tree of $G$, such that $T\cap \gamma$ is a maximal forest, i.e. intersects
every connected component of $\gamma$ in a spanning tree. With the notation of section
\ref{sec:gener-perm}, we can write this as
\begin{equation*}
  h^{1}(\gamma) = |S| - |E^{c}_{\gamma}\cap S|= (r^{*}_{G})^{\#}(\gamma), 
\end{equation*}
where $S=E_{G}\backslash T$ and $r^{*}_{G}$ is the rank function of the dual graph matroid
$M^{*}(G)$, whose bases are the complements of spanning trees. This shows that
\begin{equation*}
  h^{1}:2^{E_{G}}\rightarrow \mathbb{Z}
\end{equation*}
is supermodular.

Let us also call a pair $(T,i)$ of a spanning tree $T\subset G$ and an edge $i\in E_{G}$
\emph{admissible} if either $i\in E^{M}_{G}$ or $i\in T$ and both connected components of $T
\backslash i$ contain external momenta. Hence a pair $(T,i)$ is admissible if and only if the
monomial $\alpha_{i}\prod_{j\notin T}\alpha_{j}$ appears in $\Phi_{G}$. A subgraph $\gamma\subset G$
is then mass-momentum spanning if, for every admissible pair $(T,i)$, either $i\in \gamma$ or $T\cap
\gamma$ is not a maximal forest in $\gamma$.

\begin{prop}
  The function $s_{G}$ is supermodular.
\end{prop}

\begin{proof}
  Since we know from the above discussion that $2h^{1}$ is supermodular, we only need to show that
  \begin{equation*}
    s_{G}(\gamma_{1}) + s_{G}(\gamma) \le s_{G}(\gamma_{1}\cup \gamma_{2}) + s_{G}(\gamma_{1}\cap \gamma_{2}),
  \end{equation*}
  where $\gamma_{1}$ and $\gamma_{2}$ (and therefore $\gamma_{1}\cup\gamma_{2}$) are mass-momentum
  spanning. We can also assume that
  $h^{1}(\gamma_{1})+h^{1}(\gamma_{2})=h^{1}(\gamma_{1}\cup\gamma_{2})+h^{1}(\gamma_{1}\cap
  \gamma_{2})$, since the inequality is trivial otherwise and we can reduce to the case
  $\gamma_{1}\cup\gamma_{2}=G$. We will show that $\gamma_{1}\cap\gamma_{2}$ is also mass-momentum
  spanning. For this we must show that for every admissible pair $(T,i)$, where $T\subset
  \gamma_{1}\cup \gamma_{2}$ is a spanning tree such that $T\cap \gamma_{1}\cap \gamma_{2}$ is a
  maximal forest, we must have $i\in \gamma_{1}\cap\gamma_{2}$. From the inclusion-exclusion
  principle we have
  \begin{align*}
    h^{1}_{G}(\gamma_{1}\cup \gamma_{2}) + h^{1}_{G}(\gamma_{1}\cap\gamma_{2})
    &= |E_{\gamma_{1}\cup\gamma_{2}}| + |E_{\gamma_{1}\cap\gamma_{2}}| - |E_{(\gamma_{1}\cup\gamma_{2})\cap T}| - |E_{\gamma_{1}\cap\gamma_{2}\cap T}| \\
    &= |E_{\gamma_{1}}| + |E_{\gamma_{2}}| - |E_{\gamma_{1}}\cap E_{T}| - |E_{\gamma_{2}}\cap E_{T}| \\
    &\le h^{1}(\gamma_{1}) + h^{1}(\gamma_{2})
  \end{align*}
  with equality only if $T\cap \gamma_{1}$ and $T\cap \gamma_{2}$ are maximal forests. Since $(T,i)$
  is an admissible pair, we must have $i\in\gamma_{1}$ and $i\in\gamma_{2}$ since both graphs are
  mass-momentum spanning. But this means that $i\in \gamma_{1}\cap \gamma_{2}$ and $\gamma_{1}\cap
  \gamma_{2}$ must also be mass-momentum spanning.
\end{proof}

Our conventions on the kinematics of sub- and quotient graphs are justified by the following.

\begin{lem}
  The restrictions and contractions of $s_{G}$ by a subgraph $\gamma\subsetneq G$ are given by
  \begin{align*}
    s_{G}\vert_{\gamma} = s_{\gamma}, \quad s_{G}/_{\gamma} = s_{G/\gamma}.
  \end{align*}
\end{lem}

\begin{proof}
  The equality $s_{G}\vert_{\gamma}=s_{\gamma}$ for the restrictions follows immediately from the
  definitions. For $\eta\subset G/\gamma$, let $\tilde \eta$ be the edge subgraph corresponding to
  $E_{\gamma}\cup E_{\eta}$. The contraction equality then claims that
  \begin{align*}
    s_{G}/_{\gamma}(\eta) &= 2h^{1}(\tilde \eta) -2h^{1}(\gamma) + \delta^{mm}_{\eta}-\delta^{mm}_{\gamma}\\
                       &= 2h^{1}(\eta) + \delta^{mm}_{\eta} = s_{G/\gamma}(\eta).
  \end{align*}
  The equality $h^{1}(\tilde \eta)-h^{1}(\gamma)=h^{1}(\eta)$ follows from Example
  \ref{exam:matroid-restr-contract} and the matroid equality $M^{*}(G)\vert_{E^{c}_{\gamma}} = M^{*}(G/\gamma)$.

  We then only have to prove that $\tilde\eta$ is mass-momentum spanning if $\eta$ is. Clearly
  $\eta$ contains all massive edges of $G/\gamma$ if and only if $\tilde \eta$ contains all masses
  of $G$. Similarly, if $\tilde \eta$ contains all external vertices in a single connected
  component, then so does its contraction $\eta$. On the other hand, if all external vertices of
  $G/\gamma$ lie in the component $\eta^{0}\subseteq\eta$ and $\gamma_{1},\ldots,\gamma_{s}\subseteq
  \gamma$ are the components of $\gamma$ containing external vertices, then the subgraph defined by
  the edge set $E_{\eta^{0}}\cup E_{\gamma_{1}}\ldots\cup E_{\gamma_{s}}$ is connected and contains
  all external vertices.
\end{proof}

\begin{thm}\label{thm:feynman-supermodular}
  The Feynman polytope $P_{G}$ is the generalized permutahedron associated to the function $s_{G}$.
\end{thm}

\begin{proof}
  It is immediate from Prop. \ref{prop:factor}, that $P_{G}\subseteq P(s_{G})$. We will prove the
  reverse inclusion by induction over the number of edges. Hence we can assume that
  $P_{\gamma}=P(s_{\gamma})$ and $P_{G/\gamma}=P(s_{G/\gamma})$ for all nontrivial edge subgraphs
  $\gamma\subsetneq G$. The previous lemma then shows that
  \begin{equation*}
    F_{e^{\gamma}}P_{G} = P_{\gamma}\times P_{G/\gamma} = P(s_{\gamma})\times P(s_{G/\gamma})=F_{e^{\gamma}}P(s_{G}).
  \end{equation*}
  Since all vertices of $P(s_{G})$ are given by intersections of the above faces, we must have
  $P(s_{G})(0)\subset P_{G}$ and thus $P_{G}=P(s_{G})$.
\end{proof}

\begin{remark}
  The above factorization formula break down for some special kinematic configurations, e.g. when
  some combination of momenta $q_{F}$ are on-shell, i.e. when $q_{F}^{2}=0$ but $q_{F}\neq 0$. In
  these cases, $P(\Phi_{G})$ is not always a generalized permutahedron. An example is the box
  integral with all external momenta on-shell, which is discussed in \cite{Panzer:2015ida}.
\end{remark}

Let us recall that a graph $G$ is called \emph{1-vertex reducible} if the removal of any vertex
disconnects the graph and \emph{1-vertex irreducible (1VI)} otherwise. We consider graphs with a
single edge to be 1VI and disconnected graphs to be 1-vertex reducible. A graph on two vertices is
1VI if it contains no self-loops. Any graph then has a unique decomposition into 1VI-subgraphs. Note
that these are exactly the connected components of the graph matroid $M(G)$ and its dual $M^{*}(G)$
(See e.g. \cite[Section 2.3]{oxley2006matroid}). In particular we have that $\gamma\cap G$ is a
union of $1VI$ components if and only if the intersection with every spanning tree $T\subset G$ is a
maximal forest. By Prop. \ref{prop:factor}, this is equivalent to $R^{\psi}_{G|\gamma}=0$.

Following \cite{Smirnov_2012}, we will call a Feynman graph with generic euclidean kinematics $G$
\emph{s-irreducible}, if every 1VI-component $\gamma$ has nontrivial kinematics. This means that
either $\gamma$ contains massive edges, or its removal would disconnect the graphs into two
components, each containing external vertices. This is compatible with the notion of irreducibilty
defined in section \ref{sec:gener-perm}:

\begin{prop}\label{prop:polytope-dimension}
  Let $G$ be a connected Feynman graph with generically euclidean kinematics. The Feynman polytope
  $P_{G}$ is an irreducible generalized permutahedron if and only if $G$ is s-irreducible.
\end{prop}

\begin{proof}
  Suppose first that $G$ is s-irreducible. The generalized polyhedron $P_{G}$ is reducible if and
  only if it has dimension less than $|E_{G}|-1$. By Prop. \ref{prop:factor}, this is equivalent to
  $R^{\psi}_{G|\gamma}=R^{\Phi}_{G|\gamma}=0$ for some $\gamma\subsetneq G$. Then $\gamma$ must be a
  union of 1VI-components. If $\gamma$ were mass-momentum spanning, then its complement would be a
  union of 1VI-components with no kinematics and $G$ would not be s-irreducible. For $\gamma$ not
  mass-momentum spanning, Prop. \ref{prop:factor} gives
  \begin{equation*}
    R^{\Phi}_{G|\gamma} = R_{G|\gamma}^{\varphi} + \psi_{G|\gamma}\left( \sum_{e\in E^{M}_{G}\cap E_{\gamma}}m^{2}_{e}\alpha_{e} \right).
  \end{equation*}
  Since every 1VI component of $\gamma$ contains non-trivial kinematics, we must have either
  $R^{\varphi}_{G|\gamma}\neq 0$ or $ \sum_{e\in E^{M}_{G}\cap E_{\gamma}}m^{2}_{e}\alpha_{e}\neq
  0$, so that $R^{\Phi}_{G|\gamma}$ can not vanish.

  On the other hand, suppose $\gamma\subsetneq G$ is a 1VI component with no kinematics. Then
  $R^{\psi}_{G|\gamma}=R^{\varphi}_{G|\gamma}=0$ and the above formula shows
  $R^{\Phi}_{G|\gamma}=0$. Hence $P_{G}$ is reducible.
\end{proof}

Now suppose $G$ is s-irreducible. Let $\mathcal{F}_{G}$ denote the set of all edge subgraphs
$\gamma\subsetneq G$, such $\gamma$ and $G/\gamma$ are both $s$-irreducible, when given the
kinematics described in the beginning of this section. $\mathcal{F}_{G}$ is the disjoint union of
the two subsets
\begin{align*}
  \mathcal{S}_{G} &=\{\gamma\subsetneq G \ \rvert\ \gamma \text{ is $s$-irreducible and m.m, } G/\gamma \text{ is irreducible} \} \\
  \mathcal{H}_{G} &=\{\gamma\subsetneq G \ \rvert\ \gamma \text{ is irreducible and not m.m, } G/\gamma \text{ is }s\text{-irreducible} \}.
\end{align*}
By Prop. \ref{col:supermodular-facets}, these are exactly the facets of $P_{G}$.

\begin{col}
  Let $G$ be an $s$-irreducible Feynman graph. Then the polytope $P_{G}$ has the facet presentation
  \begin{align*}
    P_{G} = \{\langle m, e^{E_{G}} \rangle = 2h^{1}(G)+1\}\bigcap_{\gamma\in \mathcal{F}_{G}} \{\langle m, e^{\gamma} \rangle \ge 2h^{1}(\gamma)+\delta^{m.m}_{\gamma}\}.
  \end{align*}
\end{col}

\begin{remark}
  In the terminology of \cite{Speer:1975dc}, $s$-irreducible, mass-momentum spanning subgraphs
  $\gamma\subset G$ are called links and a subgraph $\gamma$ whose quotient $G/\gamma$ is
  $s$-irreducible is called saturated.
\end{remark}

\paragraph{Building sets and sectors.}

We continue to assume that $G$ is an s-irreducible Feynman graph. We want to construct smooth
refinements of the normal fan $\Sigma_{P_{G}}$. These correspond to sector decompositions by
Corollary \ref{col:sec-decomp}. Let first $\Sigma_{Hepp}$ be the normal fan of the permutahedron
$\pi_{E_{G}}$ on the set of edges $E_{G}$. By Proposition \ref{thm:feynman-supermodular} we have the
following.

\begin{prop}
  The fan $\Sigma_{Hepp}$ is a smooth refinement of $\Sigma_{P_{G}}$.
\end{prop}

Each maximal cone $\sigma\in\Sigma_{Hepp}$ is given by a complete flag $\{\emptyset=I_{0}\subsetneq
\ldots I_{n}=E_{G}\}$, or equivalently by a total order $E_{G}=\{i_{1}<\ldots<i_{n}\}$. The
coordinates of $\sigma$ are then defined by
\begin{equation*}
  \alpha_{i} = \prod_{j\le i}x_{j}.
\end{equation*}
The corresponding sectors are the classical Hepp sectors. Thus they are easy to describe but grow
superexponentially with the number of edges.

Let us apply the theory of section \ref{sec:gener-perm} to construct more economical refinements.
First consider the subset system
\begin{equation*}
  \mathcal{G}_{s} = \{E_{\gamma} \ \rvert\ \gamma\subsetneq G \text{ is s-irreducible}\}.
\end{equation*}
A special case of Prop. \ref{prop:supermodular-wonderful-refinement} then gives
\begin{prop}
  The set $\mathcal{G}_{s}$ is a building set. The corresponding fan $\Sigma_{\mathcal{G}_{s}}$ is a
  smooth refinement of $\Sigma_{P_{G}}$.
\end{prop}

\begin{remark}
  The sectors corresponding to $\Sigma_{s_{G}}$ are the Smirnov-Speer sectors considered in
  (\cite{Smirnov_2012}, \cite{Smirnov_2009}).
\end{remark}

Another possibility was recently introduced in \cite{Brown_2017}. Let us call a subgraph
$\gamma\subseteq G$ \emph{motic} if
\begin{equation*}
  s_{G}(\gamma \backslash i) < s_{G}(\gamma)
\end{equation*}
for all edges $i\in E_{\gamma}$, i.e. deleting an edge either drops the loop number, or destroys the
property of being mass-momentum spanning. Note that for massive Feynman graphs, the motic subgraphs
are exactly the disjoint unions of one-particle irreducible (1PI) graphs. Let $B_{motic}$ be the set
of motic subgraphs and
\begin{equation*}
  \mathcal{G}_{motic} = B_{motic}\cup\{\{i\} \ \rvert\ i\in E_{G}\}.
\end{equation*}

\begin{prop}
  The set $\mathcal{G}_{motic}$ is a building set and the corresponding fan $\Sigma_{motic}$ is a
  smooth refinement of $\Sigma_{P_{G}}$.
\end{prop}

\begin{proof}
  By \cite[Thm. 3.6]{Brown_2017}, the union of two motic subgraphs is again motic. Hence
  $\mathcal{G}_{motic}$ is a building set. Since $s_{G}(\gamma \backslash i)=s_{G}(\gamma)$ implies
  that $s_{\gamma}=s_{G}\vert_{\gamma}$ is reducible, we must have $\mathcal{G}_{s}\subseteq
  \mathcal{G}_{motic}$. Example \ref{exam:building-set} then shows that $\Sigma_{motic}$ refines
  $\Sigma_{s}$ and hence $\Sigma_{P_{G}}$.
\end{proof}
\begin{remark}
  The results of section \ref{sec:TV:iterblowup} show that the toric variety associated to
  $\Sigma_{motic}$ is $P^{B_{motic}}$, the iterated blowup constructed by Brown in
  \cite{Brown_2017}. Its 1PI variant was earlier introduced by Bloch-Esnault-Kreimer
  (\cite{Bloch_2006}).
\end{remark}

Let us also mention the original construction of Speer \cite{Speer:1975dc}. To describe his sectors,
we identify the set $\mathcal{S}_{G}$ of mass-momentum spanning facets with the corresponding
quotient graphs:
\begin{equation*}
  \mathcal{Q}_{G} = \{q = G/\gamma \ \rvert\ \gamma\in\mathcal{S}_{G}\}.
\end{equation*}

For an irreducible Feynman graph $G$ with generic euclidean kinematics, Speer defined a collection
of sub- and quotient graphs $\mathcal{I}\subseteq \mathcal{H}_{G}\cup \mathcal{Q}_{G}\cup \{G\}$
called s-families.

By definition, these families satisfy $G\in\mathcal{I}$ and if
$\Gamma_{1},\Gamma_{2}\in\mathcal{I}$, then either
$E_{\Gamma_{1}}\subset E_{\Gamma_{2}}$, $E_{\Gamma_{2}}\subset E_{\Gamma_{1}}$ or $E_{\Gamma_{1}}\cap E_{\Gamma_{2}}=\emptyset$.
We refer to \cite{Speer:1975dc} for the (rather involved) complete definition. The key
results of $(\cite{Speer:1975dc})$ can be summarized as follows.
\begin{thm}
  Each s-family $\mathcal{I}$ has the following properties:
  \begin{enumerate}
  \item For each $\Gamma\in\mathcal{I}$, the set
    \begin{equation*}
      E_{\Gamma} \backslash \bigcup_{\tilde\Gamma\in\mathcal{I}, E_{\tilde\Gamma}\subsetneq E_{\Gamma}}E_{\tilde\Gamma}
    \end{equation*}
    consists of precisely one element $\beta(\Gamma)$. The map $\beta:\mathcal{I}\rightarrow E_{G}$
    is a bijection.
  \item There is an admissible pair $(T,i)$ of $G$ adapted do $\mathcal{I}$, such that:
    \begin{itemize}
    \item $i\notin E_{\Gamma}$ for all $\Gamma\in\mathcal{I}$.
    \item $T\cap \gamma$ is a maximal forest for each $\gamma\in\mathcal{H}_{G}\cap \mathcal{I}$.
    \item $T/T\cap \gamma$ is a spanning tree for each $q=G/\gamma\in\mathcal{Q}_{G}\cap
      \mathcal{I}$.
    \end{itemize}
  \end{enumerate}

  To each s-family, associate the Speer sector $D_{\mathcal{I}}\subset P^{E_{G}}(\mathbb{R}^{+})$
  defined by the inequalities
  \begin{equation*}
    \max_{\substack{\gamma\in\mathcal{H}_{G}\cap \mathcal{I}\\E_{\gamma}\subset E_{\Gamma}}}\alpha_{\beta(\gamma)}\le\alpha_{\beta(\Gamma)}\le \min_{\substack{q\in\mathcal{Q}_{G}\cap \mathcal{I}\\E_{q}\subset E_{\Gamma}}}\alpha_{\beta(q)}
  \end{equation*}
  for all $\gamma\in \mathcal{H}_{G}\cap \mathcal{I}$ and $q\in\mathcal{Q}_{G}\cap \mathcal{I}$.
  Then the sectors $D(\mathcal{I})$ for different s-families cover $P^{E_{G}}(\mathbb{R}^{+})$ and
  intersect in sets of measure zero.
\end{thm}

For a quotient graph $q=G/\gamma$, we define $e^{q}=-e^{E_{q}}$. Note that in $N_{E_{G}}$, we have
the equality
\begin{equation*}
  [e^{\gamma}] = [e^{q}].
\end{equation*}
We can then rephrase Speer's result in terms of toric geometric as follows:

\begin{prop}
  For $\mathcal{I}\subset\mathcal{F}_{G}\cup \{E_{G}\}$, define the cone
  \begin{equation*}
    \Sigma_{\mathcal{I}} = \pos([e^{\Gamma}] \ \rvert\ \Gamma\in\mathcal{I}\backslash \{G\}).
  \end{equation*}
  then the cones $\{\sigma_{\mathcal{I}} \ \rvert\ \mathcal{I} \text{ an s-family}\}$ are the
  maximal cones of a smooth fan $\Sigma_{Speer}$, which refines the normal fan of $P_{G}$.
\end{prop}

\begin{proof}
  Let us first prove that the cones $\sigma_{\mathcal{I}}$ are smooth. Since $G\in\mathcal{I}$, this
  is equivalent to showing that
  \begin{equation*}
    \tilde \sigma_{\mathcal{I}} = \pos(e^{\Gamma} \ \rvert\ \Gamma\in\mathcal{I}).
  \end{equation*}
  is a smooth cone of $\mathbb{Z}^{E_{G}}$, i.e. the vectors $(e^{\Gamma}\ \rvert\
  \Gamma\in\mathcal{I})$ form a $\mathbb{Z}$-basis.

  We can then adapt the proof of \cite[Prop. 2]{Feichtner_2004_2}: Choose a linear order
  \begin{equation*}
    \mathcal{I}=\{\Gamma_{1}<\ldots<\Gamma_{E-1}<\Gamma_{E}\},
  \end{equation*}
  refining the natural order $\Gamma\preceq\Gamma'\Leftrightarrow E_{\Gamma}\subset E_{\Gamma'}$ by
  edge-inclusion and let $E_{G}=\{i_{1}<\ldots<i_{E}\}$ be the order induced by the bijection
  $\beta:\mathcal{I}\rightarrow E_{G}$. By construction of $\beta$, we have
  \begin{equation*}
    e^{\Gamma_{1}}\wedge \ldots \wedge e^{\Gamma_{r}} = \pm e^{\Gamma_{1}}\wedge \ldots \wedge e^{\Gamma_{r-1}}\ldots\wedge e^{i_{r}},
  \end{equation*}
  and an obvious induction gives
  \begin{equation*}
    e^{\Gamma_{1}}\wedge \ldots \wedge e^{\Gamma_{E}} = \pm e^{i_{1}}\wedge \ldots \wedge e^{i_{E}},
  \end{equation*}
  which shows that the $e^{\Gamma_{i}}$ form a $\mathbb{Z}$-basis.

  The coordinates $x_{\Gamma}$ of $\sigma_{\mathcal{I}}$ can then be described by
  \begin{equation*}
    \alpha_{i} = \prod_{i\in\gamma\in {H}_{G}\cap \mathcal{I}}x_{\gamma}\prod_{i\notin q\in{Q}_{G}\cap \mathcal{I}}x_{q},
  \end{equation*}
  In these coordinates, the Speer sector corresponding to $\mathcal{I}$ is described by $0\le
  x_{\Gamma}\le 1$ for $\Gamma\in\mathcal{I}$. Applying the logarithm map from Prop.
  \ref{prop:sec-decom} then shows that the cones generate a smooth complete fan.

  Let $(T,i)$ be an admissible pair, which is adapted to the s-family $\mathcal{I}$ as above. Note
  that $T/T\cap \gamma$ is a spanning tree in $q=G/\gamma$ if and only if $\gamma\cap T$ is a maximal
  forest. To $(T,i)$ corresponds the point $m=2e_{^{E_{G}\backslash T}}+e_{i}$ of the Feynman polytope
  $P_{G}$ and we have
  \begin{align*}
    \langle m, e^{\gamma} \rangle &= 2h^{1}(\gamma)=s_{G}(\gamma),\quad \gamma\in\mathcal{H}_{G}\cap\mathcal{I}\\
    \langle m, e^{\gamma} \rangle &= 2h^{1}(\gamma)+1=s_{G}(\gamma),\quad G/\gamma\in\mathcal{Q}_{G}\cap\mathcal{I}.
  \end{align*}
  Since $[e^{G/\gamma}]=[e^{\gamma}]$, we obtain from Prop. \ref{prop:maximal-cones}, that
  $\Sigma_{Speer}$ refines $\Sigma_{P_{G}}$.
\end{proof}

The Speer sectors are very economical but they are quite difficult to understand. It would be useful
to have a generalization of Speer's construction which works for all generalized permutahedra. We
conjecture the following:

\begin{con}
  Let $\overline{GP}_{E}$ be the set of generalized permutahedra on a set $E$ up to normal
  equivalence and $\overline{SGP}_{E}$ the subset consisting of polytopes whose connected
  components are simple. Then there is a natural map
  \begin{equation*}
    \overline{GP}_{E}\rightarrow \overline{SGP}_{E},\quad P\mapsto P^{s},
  \end{equation*}
  which commutes with contraction and restriction and such that $\Sigma_{P}(1)=\Sigma_{P^{s}}(1)$.
\end{con}

\begin{remark}
  In the language of \cite{aguiar17:hopf}, we ask for a (necessarily idempotent) morphism of Hopf
  monoids $\overline{GP}_{E}\rightarrow \overline{SGP}_{E}$. If we drop the condition that
  $\Sigma_{P}(1)=\Sigma_{P^{s}}(1)$, then mapping $P=P(z)$ to the polytope of its building set of
  irreducibility components $\mathcal{G}_{z}$ provides an example of such a morphism.
\end{remark}

\paragraph{Dimensional regularization.}

We can now use the results of section \ref{sec:gener-mell-transf} to define the dimensional
regularization of a Feynman integral
\begin{equation*}
  I_{G}(\lambda,D,q,m) = \Gamma(\omega_{G})\int_{P^{E_{G}}(\mathbb{R}^{+})}\prod_{e\in E_{G}}\frac{\alpha^{\lambda_{e}-1}}{\Gamma(\lambda_{e})}\left( \frac{\psi_{G}}{\Phi_{G}} \right)^{\omega_{G}}\frac{\Omega_{P^{E_{G}}}}{\psi_{G}^{D/2}}.
\end{equation*}
Recall that
\begin{equation*}
  \omega_{G} = \sum_{i\in E_{G}}\lambda_{i} - \frac{D}{2}h^{1}(\gamma).
\end{equation*}

For $\gamma$ a sub- or quotient graph of $G$, we define similarly
\begin{equation*}
  \omega_{\gamma} = \sum_{i\in E_{\gamma}}\lambda_{i} - \frac{D}{2}h^{1}(\gamma).
\end{equation*}

The convergence domain of $I_{G}$ for a graph with generic euclidean kinematics can then be
calculated as follows:
\begin{prop}
  Suppose the graph $G$ has generic euclidean kinematics. Then the Feynman Integral
  $I_{G}(\lambda,D,q,m)$ has convergence domain
  \begin{equation*}
    \Lambda_{G}=\{(\lambda,D)\in \mathbb{C}_{E_{G}}\times \mathbb{C}\ \rvert\
    \omega_{\gamma} > 0, \gamma\in \mathcal{H}_{G}, \omega_{G/\gamma}<0, \gamma\in \mathcal{S}_{G}\}
  \end{equation*}
  This domain is nonempty if and only if $G$ is s-irreducible.
\end{prop}

\begin{proof}
  Combining Theorem \ref{thm:mellin-convergence} with Proposition \ref{prop:polytope-dimension}
  shows that $(\lambda,D)\in \Lambda_{G}$ if and only if
  \begin{align*}
    \langle \lambda, e^{\gamma} \rangle - \frac{D}{2}h^{1}(\gamma) &> 0, \quad \gamma\in\mathcal{H}_{G} \\
    \langle \lambda, e^{\gamma} \rangle - \frac{D}{2}h^{1}(\gamma) - \omega_{G} &> 0, \quad \gamma\in\mathcal{S}_{G}.
  \end{align*}
  Since $\omega_{G}=\omega_{\gamma}+\omega_{G/\gamma}$ for every subgraph $\gamma$, this is
  equivalent to
  \begin{align*}
    \omega_{\gamma} &> 0, \quad \gamma\in\mathcal{H}_{G} \\
    - \omega_{G/\gamma} &> 0, \quad \gamma\in\mathcal{S}_{G}.
  \end{align*}
  This domain is nonempty if and only if $P_{G}$ has dimension $|E_{G}|-1$, which is equivalent to
  s-irreducibility by Proposition \ref{prop:polytope-dimension}.
\end{proof}

\begin{remark}
  That $\Gamma_{G}$ is nonempty for $G$ 1VI was originally proven by Speer \cite{Speer:1975dc}. The
  extension to s-irreducible graphs seems to be well-known. Another proof can found in
  \cite{Smirnov_2012}.
\end{remark}

The proposition shows that singularities corresponding to mass-momentum spanning subgraphs are more
closely associated to the quotient graphs. It will then be convenient to set
\begin{align*}
  \tilde \omega_{\gamma} =
  \begin{cases}
    -\omega_{G/\gamma}, \quad& \text{ if } \gamma \text{ is mass-momentum spanning}\\
    \omega_{\gamma}, \quad& \text{ otherwise}
  \end{cases}.
\end{align*}

In dimensional regularization, one keeps the analytic parameters $\lambda\in \mathbb{C}^{E_{G}}$
fixed (usually at integer values) and tries to expand the above integral in a Laurent series around
a point $D_{0}\in \mathbb{N}$ of the spacetime dimension. There are essentially three different
procedures to achieve this in the literature:
\begin{enumerate}
\item In the classical approach to dimensional regularization (\cite{Collins},
  \cite{_t_Hooft_1972}), the $D$-dimensional euclidean space is embedded into an
  infinite-dimensional space and the Feynman integral is split into a finite-dimensional subspace
  containing all external momenta and its orthogonal complement. Formally integrating over this
  infinite-dimensional complement gives an expression which is naturally analytic in the dimension
  $D$.
\item In the sector decomposition approach (\cite{Smirnov_1983}, \cite{HEINRICH_2008},
  \cite{Bogner_2008}), one decomposes the integration domain into cubical sectors as in section
  \ref{sec:sect-decomp}. The $\epsilon$-expansion is then explicitly computed in each sector by a
  Taylor subtraction.
\item In the analytic continuation approach (\cite{Panzer:2015ida}, \cite{von_Manteuffel_2015})),
  one applies the integration by parts procedure of section \ref{sec:gener-mell-transf} to
  analytically continue the integrals into the domain of absolute convergence.
\end{enumerate}

To our knowledge, there is no mathematical rigorous construction of the first approach. On the other
hand, the second and third fit very naturally into the framework we have developed so far.

Let us start with the sector decomposition approach. Let $\Sigma$ denote one of the smooth
refinements of $\Sigma_{P_{G}}$ constructed in the last section. From Corollary
\ref{col:sec-decomp}, we have the formula
\begin{equation*}
  I_{G}(\lambda,D,q,m) = \frac{\Gamma(\omega_{G})}{\prod_{e}\Gamma(\lambda_{e})}\sum_{\sigma}\int_{[0,1]^{|E_{G}|-1}} x_{\sigma}^{\lambda_{\sigma}}
  \left( \frac{\psi_{\sigma,G}(x_{\sigma})}{\Phi_{\sigma,G}(x_{\sigma})} \right)^{\omega_{G}}\psi_{\sigma,G}^{-D/2} \df x_{\sigma}.
\end{equation*}
The sum runs over smooth, maximal cones $\sigma=pos(e^{\gamma}\ \rvert\ \gamma\in
\mathcal{I}_{\sigma})\in\Sigma(|E_{G}|-1)$, where $\mathcal{I}_{G}\subset
2^{E_{G}}\backslash\{E_{G}\}$ is a collection of subgraphs
and $x_{\sigma}= (x_{\gamma} \ \rvert\ \gamma\in\mathcal{I}_{\sigma})$ are the associated
coordinates. The leading monomial $x_{\sigma}^{\tilde \lambda}$ in the sector $\sigma$ is given by
\begin{equation*}
  x^{\lambda_{\sigma}}_{\sigma} = \prod_{\gamma\in\mathcal{I}_{\sigma}}x_{\gamma}^{\tilde\omega_{\gamma}-1}.
\end{equation*}
The polynomials $\psi_{\sigma},\Phi_{\sigma}$ are obtained as
\begin{equation*}
  \psi_{\sigma,G}(x_{\sigma}) = x_{\sigma}^{-h^{1}(\gamma)}\psi_{G}(x_{\sigma}), \quad \Phi_{\sigma,G}(x_{\sigma}) = x^{-h^{1}(\gamma)-\delta^{mm}_{\gamma}}_{\sigma}\Phi_{G}(x_{\sigma}),
\end{equation*}
and are regular and non-vanishing on the sector $[0,1]^{E_{G}-1}$ defined by $\sigma$. Fix
$\lambda^{0}\in \mathbb{Z}^{E_{G}}$ and $D_{0}\in \mathbb{Z}$ and let $\tilde
\omega_{\gamma}^{0}=\tilde \omega_{\gamma}\vert_{\lambda=\lambda^{0},D=D_{0}}$.

Define the multi-index $\alpha_{\sigma}\in \mathbb{N}^{\mathcal{I}_{\sigma}}$ by
$\alpha_{\sigma}(\gamma) = \max(0,-\tilde \omega_{\gamma}^{0}+1)$. Let
\begin{equation*}
  F_{\sigma}(x_{\sigma}) = \left( \frac{\psi_{\sigma,G}(x_{\sigma})}{\Phi_{\sigma,G}(x_{\sigma})} \right)^{\omega_{G}}\psi_{\sigma,G}^{-D/2}
\end{equation*}
and consider its Taylor expansion in the $x_{\sigma}$ variables up to $\alpha_{\sigma}$:
\begin{equation*}
  F_{\sigma}(x_{\sigma}) = \sum_{\substack{\beta\in \mathbb{N}^{\mathcal{I_{\sigma}}}\\\beta\preceq\alpha_{\sigma}}}\frac{\partial^{\beta} F(0)}{\beta!}x^{\beta}_{\sigma} + \tilde F(x_{\sigma}).
\end{equation*}
The integral over the polynomial part can be explicitly calculated as a rational function in
$\epsilon$:
\begin{equation*}
  \int_{[0,1]^{E_{G}-1}}x_{\sigma}^{\lambda_{\sigma}}x^{\beta}\frac{1}{\beta!}\partial^{\beta}F_{\sigma}(0)\df x_{\sigma} =
  \frac{1}{\beta!}\partial^{\beta}F_{\sigma}(0)\prod_{\gamma\in\mathcal{I}_{\sigma}}\frac{1}{\tilde \omega_{\gamma}+\beta_{\gamma}}.
\end{equation*}
The integral over $\tilde F(x_{\sigma})$ is analytic around $\epsilon=0$ by construction and can be
expanded as a power series in $\epsilon$.

For Feynman graphs which are not s-irreducible, it is conventional to set $I_{G}=0$ in dimensional
regularization. On the other hand, we have seen that the sector decomposition approach still
provides an $\epsilon$-expansion in this case. Luckily, the the two prescriptions agree.
\begin{prop}
  If $G$ is not s-irreducible, then regularization by sector decomposition gives
  $I_{G}(\lambda,D,q,m)=0$.
\end{prop}

\begin{proof}
  Let us first show that the sector decomposition value is independent of the choice of refinement.
  If $\tilde\Sigma$ is another smooth fan, which refines $\Sigma$ (and thus $\Sigma_{P_{G}}$), then
  every maximal cone $\sigma\in\Sigma$ is a union of cones $\sigma_{i}\in\tilde\Sigma$, which only
  overlap in common faces. Thus the sector corresponding to $\sigma$ is the union of the sectors
  corresponding to $\sigma_{k}$. By analytic continuation, the sum over the $\sigma_{k}$-sectors
  must equal the contribution of the $\sigma$-sector. Thus the value of $I_{G}$ computed with
  respect to $\Sigma$ or $\tilde \Sigma$ are the same.

  If now $\Sigma'$ is any other smooth fan which refines $\Sigma_{P_{G}}$, then we can always find a
  smooth fan which refines both $\Sigma$ and $\Sigma'$ and which gives the same value for $I_{G}$.
  Thus the value of $I_{G}$ is independent of the choice of $\Sigma$.

  Suppose $\gamma\subset G$ is a 1VI component with no external kinematics. Choose smooth
  refinements $\Sigma_{1},\Sigma_{2}$ of the normal fans
  $\Sigma_{P_{\gamma}},\Sigma_{P_{G/\gamma}}$. We have the exact sequence of lattices
  \begin{center}
    \begin{tikzcd}
      0 \arrow{r} & \mathbb{Z}e^{\gamma} \arrow{r}{} & N_{E_G} \arrow{r}{} & N_{E_{\gamma}}\oplus
      N_{E_{G/\gamma}} \arrow{r} & 0
    \end{tikzcd}
  \end{center}
  which has a (non-canonical) splitting $N_{E_{G}}\cong \mathbb{Z}e^{\gamma}\oplus
  N_{E_{\gamma}}\oplus N_{E_{G/\gamma}}$. Let $\Sigma_{0}=\Sigma_{P^{1}}$ be the fan on
  $\mathbb{Z}e^{\gamma}$ with maximal cones $\sigma^{\pm}=\pm\mathbb{R}^{+}e^{\gamma}$. With this
  splitting, the fan $\Sigma=\Sigma_{0}\times \Sigma_{1}\times \Sigma_{2}$ is a refinement of
  $P_{G}$. Every maximal cone of $\Sigma$ is of the form $\sigma^{\pm}=(\pm \mathbb{R}^{+})\times
  \sigma$, for $\sigma\in\Sigma_{1}\times\Sigma_{2}(|E_{G}|-2)$. Let $x_{\gamma}$ be the variable
  corresponding to $\mathbb{R}^{+}e^{\gamma}\in \Sigma_{0}$. In the variables of $\sigma^{\pm}$, the
  integrand takes the form
  \begin{equation*}
    x_{\gamma}^{\pm \tilde\omega_{\gamma}-1}x_{\sigma}^{\lambda_{\sigma}}\left( \frac{\psi_{\sigma,G}(x_{\sigma})}{\Phi_{\sigma,G}(x_{\sigma})} \right)^{\omega_{G}}\psi_{\sigma,G}^{-D/2}
    =  x_{\gamma}^{\pm \tilde\omega_{\gamma}-1}x_{\sigma}^{\lambda_{\sigma}}F_{\sigma}(x_{\sigma}),
  \end{equation*}
  where $x_{\sigma}$ denotes the variables of the cone $\sigma\in\Sigma_{1}\times\Sigma_{2}$.
  Crucially, the function $F_{\sigma}(x_{\sigma})$ does not depend on $x_{\gamma}$.

  Let $I_{\sigma}=\int_{[0,1]^{|E_{G}|-2}}x_{\sigma}^{\lambda_{\sigma}}F_{\sigma}(x_{\sigma})$. Then
  we get
  \begin{align*}
    \frac{\prod_{e}\Gamma(\lambda_{e})}{\Gamma(\omega_{G})}I_{G}(\lambda,D,q,m)
    &= \sum_{\sigma^{\pm}}I_{\sigma}\int_{[0,1]}x^{\pm\tilde\omega_{\gamma}-1}\df x_{\gamma} \\
    &= \sum_{\sigma}I_{\sigma}\left( \frac{1}{\tilde \omega_{\gamma}}-\frac{1}{\tilde\omega_{\gamma}} \right) = 0.
  \end{align*}
\end{proof}

The sector decomposition approach has the downside, that the integrals over the different sectors
usually lead to analytic functions which are much more complicated then their sum $I_{G}$. For this
reason, the recent articles (\cite{von_Manteuffel_2015}, \cite{Panzer:2015ida}) advocate calculating
the $\epsilon$-expansion by the integration by parts procedure described in section
\ref{sec:analyt-dimens-regul}. First let us combine Theorem \ref{thm:meromorphic-cont} with our
results on the Feynman polytope $P_{G}$.

\begin{thm}
  If $G$ is s-irreducible, then the amplitude $I_{G}$ can be expressed as
  \begin{equation*}
    I_{G}(\lambda,D,q,m) = \omega_{G}\left( \prod_{\gamma\in\mathcal{F}_{G}}\Gamma(\tilde \omega_{\gamma}) \right)J_{G}(\lambda,D,q,m),
  \end{equation*}
  where $J_{G}$ is analytic for all $(\lambda,D)\in \mathbb{C}^{E_{G}}\times \mathbb{C}$ and
  external momenta and masses $(q,m)$ satisfying the inequalities of definition
  \ref{defin:generic-kinematics}.
\end{thm}

We can describe the analytic continuation more concretely as follows. Let again
$a(\gamma)=\max(0,\tilde \omega^{0}_{\gamma}+1)$. For each $\gamma\in\mathcal{F}_{G}$, we integrate
by parts $a(\gamma)$-times and obtain an expression of the form
\begin{equation*}
  I(\lambda^{0},D^{0}+2\epsilon,q,m) = \sum_{\beta}L_{\beta}(\epsilon)I(\lambda^{\beta},D^{\beta}+2\epsilon,q,m):= \sum_{\beta}L_{\beta}(\epsilon)I_{\beta}(\epsilon)
\end{equation*}
where $(\lambda^{\beta},D^{\beta})\in \mathbb{Z}^{E_{G}}\times\mathbb{Z}$ are shifted values of the
analytic parameters and dimension and $L_{\beta}(\epsilon)$ are rational functions in $\epsilon$
depending polynomially on the external kinematics. By construction, each $I_{\beta}$ is now analytic
in a neighbourhood of $\epsilon$.

\begin{remark}
  The authors of \cite{von_Manteuffel_2015} remark that the partial integrations are easy to
  calculate, but the number of terms in the above sum can grow very rapidly. Note that the result
  depends on the order of partial integrations in general. The proof of Theorem
  \ref{thm:meromorphic-cont} suggests the following naive algorithm: At each stage, choose the next
  direction such that the Newton polytopes of the Laurent monomials appearing in the numerator are
  as small as possible.

  It is known that the space of integrals of the above form is finite-dimensional
  (\cite{Smirnov_2010}) and the above sum can be considerably simplified. But reducing a given
  integral to a basis of so-called ``Master'' integrals is quite difficult (See e.g.
  \cite{LAPORTA_2000},\cite{Chetyrkin_1981} for the classical IPB technique). We refer to the recent
  article (\cite{bitoun17:feynm}) for a $D$-module approach which is quite close to our toric
  viewpoint.
\end{remark}

We can now express the $\epsilon$-expansion of $I_{G}$ in terms of the homogeneous coordinates of
$X_{\Sigma}$ as follows.

\begin{thm}
  Let $G$ be s-irreducible and $\Sigma$ a smooth refinement of $\Sigma_{P_{G}}$. The functions
  $I_{\beta}(\epsilon)$ have the series expansion
  \begin{align*}
    I_{\beta}(\epsilon) = \sum_{k_{1},k_{2}=0}^{\infty}\frac{h^{1}(G)^{k_{2}}}{k_{1}!k_{2}!}\epsilon^{k_{1}+ k_{2}}
    \int_{X_{\Sigma}(\mathbb{R}^{+})}x^{\lambda^{\beta}}\left( \frac{\psi_{G}}{\Phi_{G}} \right)^{\omega^{\beta}_{G}}\psi_{G}^{-D^{\beta}/2}
    \log^{k_{1}}\left( \psi_{G} \right)\log^{k_{2}}\left( \frac{\psi_{G}}{\Phi_{G}} \right)\Omega_{X_{\Sigma}},
  \end{align*}
  where $\omega^{\beta}_{G} = \omega_{G}\vert_{\lambda=\lambda^{\beta},D=D^{\beta}}$.
\end{thm}

\begin{remark}
  Suppose the masses and momenta are generically euclidean and \emph{rational}. We can write the
  logarithmic powers appearing above as
  \begin{equation*}
    \log^{k}(h(x))=\int_{[0,1]^{k}}\prod_{i=1}^{k}\frac{h(x)-1}{(h(x)-1)t_{i}+1}\df t_{i}.
  \end{equation*}
  Inserting this relation into the above expression for $I_{\beta}$ shows that the coefficients of
  the $\epsilon$-expansion are then periods in the sense of Kontsevich-Zagier
  (\cite{Kontsevich_2001},\cite{MR3618276}), a fact that was first proven by Bogner and Weinzierl
  (\cite{Bogner_2009}) using sector decompositions.
\end{remark}

\appendix

\section{Polyhedral geometry}
\label{sec:appendix}
In this appendix, we recall some concepts from polyhedral geometry. We refer the reader to
\cite{Ziegler_1995} for further details and proofs.

Let $V$ be a finite-dimensional real vector space and $V^{*}$ its dual. A polyhedron is subset of
$V$, given as an intersection of finitely many affine halfspaces, i.e. it is a subset of the form
\begin{equation*}
  P = \{ p\in V \ \rvert\ \langle p, u_{i} \rangle \ge d_{i}, i=1,\ldots,s\},
\end{equation*}
where $u_{i}\in V^{*}$ and $d_{i}\in \mathbb{R}$. The \emph{affine span} $H_{P}$ of a polyhedron $P$
is the smallest affine subspace of $V$ containing $P$. Any weight vector $u\in V^{*}$ defines the
face
\begin{equation*}
  F_{u}P = \{ p\in P \ \rvert\ \langle p, u \rangle = \min_{\tilde p\in P}\langle \tilde p, u \rangle\}
\end{equation*}
We denote the set of dimension $k$ faces of $P$ by $P(k)$. A face $F=F_{u}P$ of codimension one is
called a \emph{facet} and a face of dimension $0$ is called a \emph{vertex}. Any polyhedron of
dimension $n$ then has an irredundant presentation
\begin{align*}
  P = \{ p\in H_{P} \ \rvert\ &\langle p, u_{F}\rangle \ge d_{F}, F\in P(n-1)\},
\end{align*}
where $u_{F}$ is the weight vector defining the facet $F$.

A \emph{ polyhedral cone} is a polyhedron of the form
\begin{equation*}
  \sigma = \{ p\in V \ \rvert\  \langle p,u_{i}  \rangle \ge 0\}.
\end{equation*}
Alternatively, we can describe $\sigma$ as the positive hull of finitely many vectors
$v_{1},\ldots,v_{r}\in V$:
\begin{equation*}
  \sigma = \pos(v_{1},\ldots,v_{r}) := \{p\in V \ \rvert\ p=\sum\lambda_{i}v_{i}, \lambda_{i}\ge 0\}.
\end{equation*}
A polyhedral cone is called \emph{strongly convex}, if it does not contain a subspace of positive
dimension, i.e. if $0\in\sigma$ is a vertex. Every polyhedral cone $\sigma\subset V$ defines a dual
cone
\begin{equation*}
  \sigma^{\vee} = \{u\in V^{*} \ \rvert\ \langle v, u \rangle\ge 0 \text{ for all } v\in \sigma\}.
\end{equation*}
A bounded polyhedron $P\subset V$ is called a \emph{polytope}. Every polytope can be described as
the convex hull of its vertices $v\in P(0)$, i.e.
\begin{equation*}
  P = \{ p\in V \ \rvert\ p=\sum_{v\in P(0)} \lambda_{v}v,\ \lambda_{v}\ge 0,\ \sum_{v\in P(0)}\lambda_{v} = 1\}.
\end{equation*}
For two polytopes $P,Q\subset V$, their Minkowski sum is defined as
\begin{equation*}
  P+Q = \{p\in V \ \rvert\  p=p_{1}+p_{2} \text{ for } p_{1}\in P, p_{2}\in Q\}.
\end{equation*}
For $r\in (0,\infty)$, we also define the scaled polytope
\begin{align*}
  rP &= \{rp\in V \ \rvert\ p\in P\}.
\end{align*}

\bibliographystyle{habbrv}
\bibliography{literature}{}
\end{document}